\documentclass[onefignum,onetabnum,oneeqnum,onethmnum]{siamart220329}

\usepackage{lipsum}
\usepackage{amsfonts}
\usepackage{graphicx}
\usepackage{epstopdf}
\usepackage{algorithmic}
\ifpdf
  \DeclareGraphicsExtensions{.eps,.pdf,.png,.jpg}
\else
  \DeclareGraphicsExtensions{.eps}
\fi

\newsiamremark{remark}{Remark}
\newsiamremark{hypothesis}{Hypothesis}
\crefname{hypothesis}{Hypothesis}{Hypotheses}
\newsiamthm{claim}{Claim}

\headers{ON TWO-HANDED CONNECTED ASSEMBLY PARTITIONING}{P. K. Agarwal,  B. Aronov, T. Geft, and D. Halperin}

\title{On Two-Handed Planar Assembly Partitioning with Connectivity Constraints\thanks{%
A previous version of this paper appeared in the proceedings of SODA 2021~\cite{DBLP:conf/soda/AgarwalAGH21}.
In this version we generalize our connected-assembly-partitioning algorithm for unit-grid squares to polygonal assemblies.
\funding{P. K. Agarwal has been partially
			supported by NSF grants IIS-18-14493 and CCF-20-07556.
   B.~Aronov has been partially supported by NSF grants CCF-15-40656 and CCF-20-08551, and by grant 2014/170 from the US-Israel Binational Science Foundation.
   Work on this paper by T.~Geft and D.~Halperin has been supported in part by the Israel Science Foundation (grant no.~1736/19), by NSF/US-Israel-BSF (grant no.~2019754), by the Israel Ministry of Science and Technology (grant no.~103129), by the Blavatnik Computer Science Research Fund, and by the Yandex Machine Learning Initiative for Machine Learning at Tel Aviv University.
   T.~Geft has also been supported by a scholarship from the Shlomo Shmeltzer Institute for Smart Transportation at Tel Aviv University.
   }}}

\author{Pankaj K. Agarwal\thanks{Department of Computer Science, Duke University, USA 
  (\email{pankaj@cs.duke.edu}).}
\and Boris Aronov\thanks{Department of Computer Science and Engineering,
			Tandon School of Engineering, New York
			University, Brooklyn, NY 11201, USA
  (\email{boris.aronov@nyu.edu}).}
\and Tzvika Geft\thanks{Blavatnik School of Computer Science, Tel-Aviv University, Israel
  (\email{zvigreg@mail.tau.ac.il}, \email{danha@tauex.tau.ac.il}).}
\and Dan Halperin\footnotemark[4]}

\usepackage{amsopn}

\ifpdf
\hypersetup{
  pdftitle={On Two-Handed Planar Assembly Partitioning with Connectivity Constraints},
  pdfauthor={P. K. Agarwal,  B. Aronov, T. Geft, and D. Halperin}
}
\fi

\usepackage{fullpage}  %
\usepackage{adjustbox}
\usepackage{mathpazo}
\usepackage{mathtools}%
\usepackage{amsfonts}
\usepackage{euscript}
\usepackage{bbold}
\usepackage{cite}
\usepackage{microtype}
\usepackage{xcolor}

\usepackage{fixltx2e}
\usepackage{xspace}

\makeatletter
\newcommand{\Xcite}[1]{{\let\cite@adjust\empty\cite{#1}}}
\makeatother

\usepackage{dblfloatfix} %
\usepackage{comment}

\graphicspath{{./figures/}{./}}
\usepackage{float}
\usepackage[linesnumbered,ruled,vlined,algo2e]{algorithm2e} %

\usepackage{enumitem, cleveref}

\usepackage [english]{babel}
\usepackage [autostyle, english = american]{csquotes}
\MakeOuterQuote{"}
\usepackage{cleveref} %

\definecolor{gray}{rgb}{0.35,0.35,0.35}
\definecolor{blue}{rgb}{0,0,1}
\definecolor{red}{rgb}{1,0,0}
\definecolor{orange}{rgb}{0.75, 0.4, 0}
\definecolor{green}{rgb}{0.0, 0.5, 0.0}

\newcommand{\newrestriction}{with neighboring variable pairs}
\newcommand{\newsat}{Monotone Planar 3-SAT with Neighboring Variable Pairs}
\newcommand{\newsatCap}{Monotone Planar 3-SAT with Neighboring Variable Pairs}
\newcommand{\newsatshort}{\textsc{MP-3SAT-NVP}}
\newcommand{\gridsq}{unit-grid square}
\newcommand{\gridass}{grid square assembly}

\newcommand{\nvars}[1]{\ensuremath{\textsc{Nbr\!Vars}(#1)}}
\newcommand{\nvarstext}{neighboring variables}
\newcommand{\rewire}[1]{\textsc{Replace}\ensuremath{(#1)}}
\newcommand{\set}[1]{\ensuremath{\{#1\}}}

\newcommand{\dirup}{the $+y$-direction}
\newcommand{\dirdown}{the $-y$-direction}

\newcommand{\clause}[3]{\ensuremath{(#1 \lor #2 \lor #3)}}
\newcommand{\clausetwo}[2]{\ensuremath{(#1 \lor #2)}}
\newcommand{\clausea}[1]{\ensuremath{a_{#1}}}
\newcommand{\clauseb}[1]{\ensuremath{b_{#1}}}
\newcommand{\clausec}[1]{\ensuremath{c_{#1}}}
\newcommand{\claused}[1]{\ensuremath{d_{#1}}}
\newcommand{\agraph}{\mathcal{G}}
\newcommand{\bgraph}{\mathcal{H}}
\newcommand{\rgraph}{\agraph_{|B}}
\newcommand{\ograph}{\rgraph^\sqcup}
\newcommand{\compl}[1]{\overline{#1}}

\newcommand{\dZ}{\mathbb{Z}}

\def\ph{\varphi}

\newcommand{\Sh}{\mathop{\mathsf{Sh}}}

\let\bd\partial

\newcommand{\true}{\textsc{true}}
\newcommand{\false}{\textsc{false}}
\newcommand{\ass}{\ensuremath{\mathcal{A}(\up{})}}

\newcommand{\up}{\ensuremath{\mathit{UP}}}
\newcommand{\down}{\ensuremath{\mathit{DOWN}}}

\newcommand{\algname}{\textsc{Augment}}
\newcommand{\subprocname}{\textsc{Connect}}

\newcommand{\nodef}[1]{\ensuremath{\mathcal{V}(#1)}}
\newcommand{\CS}{\mathop{\mathsf{CS}}}

\newcommand{\monass}{horizontally monotone}

\begin{document}

\maketitle

\begin{abstract}
Assembly planning is a fundamental problem in robotics and automation, which involves designing a sequence of motions to
bring the separate constituent parts of a product into their final
placement in the product.
Assembly planning is naturally cast as a disassembly problem, giving rise to the \emph{assembly partitioning} problem:
Given a set $A$ of parts, find a subset $S\subset A$, referred to as a subassembly, such that $S$ can be rigidly translated to 
infinity along a prescribed direction without colliding with $A\setminus S$.
While assembly partitioning is efficiently solvable, it is further desirable for the parts of a subassembly to be easily held together.
This motivates the problem that we study, called \emph{connected-assembly-partitioning}, which additionally requires each of the two subassemblies, $S$ and $A\setminus S$, to be
connected. 
We show that this problem is NP-complete, settling an open question posed by Wilson et al.~(1995) a quarter of a 
century ago,
even when $A$ consists of unit-grid squares (i.e., $A$ is polyomino-shaped).
Towards this result, we prove the NP-hardness of a new \textsc{Planar 3-SAT} variant having an adjacency requirement for variables appearing in the same clause, which may be of independent interest.
On the positive side, we give an $O(2^k n^2)$-time fixed-parameter tractable algorithm (requiring low degree polynomial-time pre-processing) for an assembly $A$ consisting of polygons in the plane, where $n=|A|$ and $k=|S|$.
We also describe
a special case of unit-grid square assemblies, where a connected partition can always be found
in $O(n)$-time.
\end{abstract}



\section{Introduction}
\label{sec:intro}

Automating assembly has been a fundamental research area in robotics since its early days.
A key problem in this area is \emph{assembly planning} (also called \emph{assembly sequencing}), where the goal is to find a sequence of motions that merges initially separated parts into their final relative positions in an assembly.
The problem's output is called an \emph{assembly sequence}.
Solving it would allow computer-aided design (CAD)
systems to provide better feedback to designers to help them create products
that are more cost-effective to manufacture~\cite{halperin2000general}.
It is convenient to study assembly planning
in reverse order, where the following key problem, \emph{assembly
partitioning}, arises: Given a set of parts in their final placement in a product, partition them into two subsets, each regarded as a rigid
body and referred to as a \emph{subassembly}, such that these two subassemblies
can be moved sufficiently far away from each other without colliding
with one another.
By recursively applying assembly partitioning, one can obtain a \emph{two-handed} assembly sequence, where exactly two rigid subassemblies are mated at each step to result in a larger one.
We study such sequences since they are ubiquitous in assembly planning.

\begin{figure}[htb]
	\centering
	\includegraphics[width=0.35\textwidth]{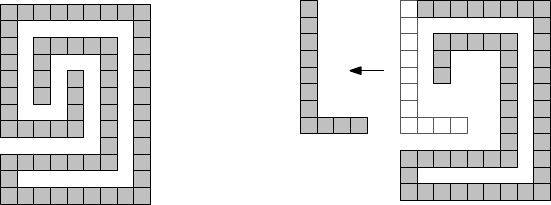}
	\caption{A negative (left) and positive (right) instance of the connected-assembly-partitioning problem. In each instance, the assembly is composed of unit-grid squares and $\vec{d}$ (the prescribed translation direction) is the $-x$-direction.
	The right instance is shown after the subassembly $S$ (which results in a connected partition) has been translated leftward.
	}
	\label{fig:intro}
\end{figure}
	
We tackle a prevalent constraint, namely \textit{connectivity}, which requires that each subassembly constructed during the sequence remains connected (such assembly sequences have also been called \textit{contact-coherent}~\cite{DBLP:conf/icra/Wolter91}).
A subassembly is \textit{connected} if the graph representing contact between parts is connected; see below for a formal definition of the graph and refer to Figure~\ref{fig:intro}.
Subassemblies that are not connected are more difficult to manipulate and mate with other subassemblies; e.g., additional fixturing might be needed.
Therefore, requiring subassemblies to be connected is a common constraint in assembly planning~\cite{DBLP:journals/ijrr/WilsonKLL95, halperin2000general, DBLP:journals/tase/GunjiDBB18, deepak2019assembly, munker2022cad}.
An approach that has been typically used to solve assembly planning with the connectivity constraint is "enumerate-and-test": partitions that result in connected subassemblies are enumerated and then tested for the geometric feasibility of separating the subassemblies~\cite{gen-and-test, thomas2003efficient, munker2022cad}. %
As the number of ways to partition an assembly into two connected subassemblies may be exponential, this approach does not have a polynomial running time guarantee. %

Although assembly planning can be performed offline for a given assembly, efficient algorithms are needed to facilitate computer-assisted design for assembly (DfA).
A common motivation here is the need to provide designers with immediate feedback.
This facilitates work in an interactive optimization cycle where small changes in geometry are made and evaluated.
For complex products, it may be hard to take the assemblability requirement into account in the early stages of the design.
Therefore, assembly planning is fundamental for optimizing not only the manufacturability of the product but also the design process itself~\cite{jimenez2013survey}.

Achieving a high level of automation in assembly planning is hindered by the multitude of practical considerations~\cite{constraints-survey} and the combinatorial explosion inherent to the underlying algorithmic problem.
In its most general form, assembly planning is PSPACE-hard~\cite{pspace}.
Yet even the solvable case of \emph{one-step translations}, i.e., where subassemblies are mated using a \emph{single} translation, allows an exponential number of sequences to arise~\cite{DBLP:conf/case/GeftTGH19}.
As a result, satisfying all the desired constraints means that manual involvement of highly-skilled users is typically required~\cite{AutoAssem, DBLP:journals/tase/GunjiDBB18}.
Such limitations and the great desire for a unified solution to assembly planning have been summarized in a recent review:
"A robust and efficient computer-aided methodology is much needed to generate optimal assembly sequences automatically from 3D CAD environment considering all the necessary assembly predicates without any user intervention"~\cite{bahubalendruni2016review}.
In this work, we advance toward this goal by simultaneously considering geometric feasibility and connectivity, which are arguably the two most fundamental constraints. %

We focus on \textit{planar assemblies}, where the parts are polygons in the plane.
For such assemblies, where we do not have the connectivity constraint,
assembly partitioning with a single translation is efficiently solvable~\cite{halperin2000general, wilson1994geometric}. On the other hand, a generalization to multiple translations and other motions was shown to make the problem NP-complete~\cite{DBLP:journals/ipl/KavrakiLW93}.
With the connectivity constraint, for multiple translations, Kavraki and Kolountzakis~\cite{kavraki1995partitioning} showed that the problem remains NP-complete.
However, the complexity of the (otherwise solvable) case of single translations has been posed as an open problem in 1995~\cite{DBLP:journals/ijrr/WilsonKLL95}.
The question has been recently reraised for the special case of polyominoes composed of unit-grid squares that arise in programmable matter~\cite{DBLP:journals/ral/SchmidtMHBF18, DBLP:conf/soda/Balanza-Martinez20}.
We settle this long-standing open question in the negative while also providing a fixed-parameter tractable algorithm for the problem.

\paragraph{Our contributions}
We study the \emph{connected-assembly-partitioning problem}, which is formally defined as follows.
An \emph{assembly} is a set $A$ of $n$ polygons with pairwise-disjoint interiors in the plane, also called \emph{parts}, 
in some given relative placement.
A \textit{subassembly} is a non-empty proper subset $S$ of the polygons composing an assembly $A$, 
in their relative placement in~$A$. 
Two polygons $P_i, P_j \in A$ are \textit{edge-connected} if $P_i\cap P_j$ contains a line segment of non-zero length.
The \textit{adjacency graph} of an assembly $A$, denoted by $\agraph(A)$, is an undirected graph that contains a node for each part in~$A$, and has an edge between two nodes if and only if the corresponding parts in $A$ are edge-connected.
For a subset $S\subseteq A$, let $\agraph(S)$ be the subgraph of $\agraph(A)$ induced by $S$.
A subassembly $S \subseteq A$ is called \textit{connected} if 
$\agraph(S)$ is connected.
Let $A$ denote an assembly and $\vec{d}$ denote a direction in the plane.
A subassembly $S \subset A$ is called a \textit{partition} of $A$ in direction $\vec{d}$ if $S$ can be rigidly translated arbitrarily far away along $\vec{d}$ without colliding with $A \setminus S$ (sliding of one subassembly along the other, namely motion in contact, is
allowed). A subassembly $S \subset A$ is called \textit{connected partition} if each of the subassemblies $S$ and $A \setminus S$ is connected (when viewed in isolation from each other). 

Given an assembly $A$ and a direction $\vec{d}$, our goal is to decide whether there is a connected partition $S \subset A$ in direction $\vec{d}$.
We refer to this problem as \textsc{Planar Partitioning into Connected Subassemblies using a Single Translation (PPCST)}.

Our first result is that \textsc{PPCST} is NP-complete even for the special case when $A$ 
is a set of unit-grid squares, i.e., $A$ is polyomino-shaped (see Section~\ref{sec:hardness-assembly} for a precise definition).
Our setting is considerably more constrained than the previously most restricted NP-complete variant of the problem, which was studied by Kavraki and Kolountzakis~\cite{kavraki1995partitioning}.
In their construction, the assembly has non-convex polygons, and the number of translations required to partition it (i.e., to bring the two subassemblies arbitrarily far apart) is linear in the number of parts.
We reduce \textsc{Monotone Planar 3-SAT}\footnote{The problem is also known as \textsc{Planar Monotone 3-SAT}.}~\cite{DBLP:conf/cocoon/BergK10}, which is known to be NP-complete, to our problem.
Our proof works in three stages.
First, we show in Section~\ref{sec:sat} that a restricted version of \textsc{Monotone Planar 3-SAT}, in which there is a total ordering of variables and each clause with three literals has two adjacent variables, is also NP-complete; we believe that this result is of independent interest.\footnote{The new SAT version was independently used to resolve several long-standing open problems in Tile Self-Assembly~\cite{DBLP:conf/icalp/CaballeroGSW22}.}
We refer to the latter problem as \textsc{\newsat} (\newsatshort).
The second stage, presented in Section~\ref{sec:hardness-assembly}, reduces the latter problem to the connected-assembly-partitioning problem.
In Section~\ref{ssec:polygon-construction} we present the reduction for the case in which each part of the assembly is a rectilinear polygon.
Then, in Section~\ref{subsec:square-construction}, we adapt this construction to the case in which each part is a \gridsq{}.

Our second result is an $O(N(n+\log N)+2^k n^2)$-time algorithm to determine whether a connected partition with $k$ 
polygons (i.e., $|S|=k$) exists; if so, the algorithm reports the subassembly (Section~\ref{sec:FPT}). Here $n=|A|$ and $N$ is the total number of vertices in $A$. From the robotic application perspective, targeting a small $k$ can be advantageous as the subassembly needs to be grasped and moved.
Such manipulations become harder as the number of parts in the subassembly grows. 
In particular, our algorithm shows that the \textsc{PPCST} problem is fixed-parameter tractable (FPT)~\cite{DBLP:books/sp/CyganFKLMPPS15} in the size of the partition. We remark that there has been considerable interest in fine-grained algorithm design for geometric problems; see, e.g.,~\cite{DBLP:conf/compgeom/AgrawalKL0Z20,DBLP:journals/talg/CabelloGKMR11,DBLP:conf/compgeom/EibenL20,DBLP:journals/comgeo/GiannopoulosKRW13} and the references therein.

In our setting we need to satisfy two requirements: (i) that the motion of one subassembly will be free from collision with the other subassembly, and (ii) that each subassembly will be connected. Satisfying each of these requirements separately is simple (for collision-free separation we use \emph{directional blocking graphs}~\cite{wilson1994geometric}).
The difficulty of the connected-assembly partitioning problem is the need to simultaneously satisfy both requirements. Our algorithm builds a partition incrementally, while alternating between the satisfaction of either requirement.
We do so using the bounded search tree method~\cite{DBLP:books/sp/CyganFKLMPPS15}.
Given a partial solution $S$, we identify at most two candidate subsets $S_1, S_2$, each containing $S$, such that any 
connected partition containing $S$ has to contain at least one of $S_1$ or $S_2$.
This gives (at most) two options of how to extend the current partial solution, which we try until we find a solution or find an invalid branch.

Finally, we describe an $O(n)$-time algorithm for a special class of unit-grid square assemblies, called \emph{horizontally monotone} assemblies, for which a connected partition always exists (Section~\ref{sec:positive_res}).
Assemblies in this class, which can be recognized in $O(n)$-time, do not have "interlocking hook" structures (such as those shown in Figure~\ref{fig:intro}, left), which allows finding a connected partition using case analysis.

\paragraph{Related work}
	Assembly planning is a well-studied problem in manufacturing and robotics.
	Early planners (i.e., algorithms for solving assembly planning) have resorted to a complete enumeration of all possible assembly operations in order to determine which of them are feasible~\cite{gen-and-test}.
	Others posed questions to a human expert in order to establish precedence between operations~\cite{user-queries}.
	Such initial efforts highlighted the critical role of efficient (automatic) geometric reasoning in assembly planning, thus giving rise to the assembly partitioning problem.
	In an early work, Arkin et al.~\cite{DBLP:conf/compgeom/ArkinCM89} relate assembly partitioning for polygons in the plane to computing monotone paths among polygonal obstacles.
	Their connection led to an efficient algorithm for one-step translations.
	Another early work presented an algorithm to efficiently compute a sequence of translations to separate two polygons~\cite{DBLP:journals/dcg/PollackSS88}.
	
	A significant advancement by Wilson and Latombe~\cite{wilson1994geometric} was the introduction of the \emph{non-directional blocking graph (NDBG)}, which allows efficiently representing all geometrically feasible partitions depending on the allowed assembly motions.
	For one-step translations and infinitesimal rigid motions, their approach avoids the inherent combinatorial trap in assembly planning and leads to polynomial time assembly partitioning algorithms in both 2D and 3D~\cite{halperin2000general}.
	Beyond addressing the feasibility question, the NDBG allows a systematic exploration of all possible partitions, which is useful for finding partitions that adhere to/optimize some criterion~\cite{DBLP:conf/case/GeftTGH19}.
	On the negative side, deciding whether a partition exists for an arbitrary number of translations (and some related variants) was shown to be NP-complete~\cite{DBLP:journals/ipl/KavrakiLW93}, even with the connectivity constraint~\cite{kavraki1995partitioning}.
	Nevertheless, one-step translations were generalized through efficient assembly partitioning algorithms for motions consisting of a constant number of translations~\cite{DBLP:journals/ar/HalperinW96, halperin2000general}.

	Given the algorithmic success in finding geometrically feasible assembly sequences, the next natural goal became finding sequences meeting additional desired properties or optimality criteria.
    Practical considerations may be connectivity, stability, and graspability of the subassemblies that arise in the assembly plan, to name just a few of a multitude of considerations~\cite{constraints-survey}.
    In turn, many optimization variants of assembly planning were examined by
	Goldwasser et al.~\cite{goldwasser1996complexity, goldwasser-thesis} for assembly sequences consisting of one-step translations, such as minimizing the number of directions used to bring parts into the assembly.
	They presented multiple hardness and inapproximability results (though with some pertaining to non-geometric generalizations of the problem), thereby highlighting the challenging computational nature of addressing more than the mere basic considerations.
	
	With these results in the backdrop, algorithmic progress in assembly planning has been slow since the highly active research period it saw in the 1990s.
    Given this slow progress, our work can be seen as a significant step forward in the development of efficient exact algorithms for assembly planning.
    A major underlying element of this advancement is the adoption of parameterized complexity, which has so far been rarely explored in assembly planning.
    Other lines of research have confronted the problem's inherent combinatorial explosion using soft computing methods, such as genetic algorithms~\cite{brodbeck2014automatic-genetic} and other heuristic methods~\cite{deepak2019assembly}.

	Recent years have seen an interest in algorithmic problems for the special case of unit-grid square assemblies.
	Motivated by programmable matter, in such problems unit squares represent tiles or particles that assemble in the micro- and nano-scale.
	In the \emph{tilt model}~\cite{DBLP:journals/algorithmica/BeckerFKKRSS20} such tiles are moved by applying a uniform global force on them.
	The model gives rise to essentially the same assembly partitioning problem that we study, namely with one-step translations and the connectivity constraint, for which several positive results were obtained.
    
    Schmidt et al.~\cite{DBLP:journals/ral/SchmidtMHBF18} give efficient connected-assembly-partitioning algorithms for grid square assemblies having convex holes (or no holes), where \emph{holes} are the bounded connected components of the assembly's complement in the plane.
    The algorithms' running times are $O(n^3 \log n)$ and $O(n^4 \log n)$ for the case of no holes and convex holes, respectively, where $n$ is the number of grid squares in the assembly.
    Building off of Schmidt et al.~\cite{DBLP:journals/ral/SchmidtMHBF18}, Balanza-Martinez et al.~\cite{DBLP:conf/soda/Balanza-Martinez20} show that a \emph{complete} assembly sequence can be found (if it exists) in $O(n^4 \log n)$-time for hole-free grid square assemblies.
    For such assemblies, they show that one may greedily find a connected partition, and recursively disassemble the resulting subassemblies.
    This is not true in general, since it is possible to make the "wrong" partition, where the resulting subassemblies cannot be further disassembled~\cite[Figure 4]{DBLP:journals/ral/ManzoorSLKKB17}.\footnote{This occurs due to the connectivity constraint. Without the constraint, it is easily seen that any partition does not hinder further disassembly.}
    
    Other efficient assembly sequencing algorithms (with the connectivity constraint) for grid square assemblies were obtained by restricting the allowed partitions in addition to the assembly's structure.
    Schmidt et al.~\cite{DBLP:journals/ral/SchmidtMHBF18} present an algorithm for assemblies with convex holes and straight cut partitions.
    For the case where one square at a time is added to the assembly, the problem is fixed-parameter tractable when parameterized by the number of holes~\cite{DBLP:journals/ral/ManzoorSLKKB17}.
    Whether this problem variant is NP-complete is unknown, though its natural generalization to 3D, with a unit cube added at a time, is NP-complete~\cite{DBLP:journals/algorithmica/BeckerFKKRSS20}.
	
	Our connected-assembly-partitioning algorithm handles more general cases than the aforementioned algorithms.
	In contrast to previous results, our algorithm places no restrictions on the partition nor on the assembly's structure, and it also applies to polygons.
	In particular, we allow arbitrarily-shaped holes in the assembly, which appear to be the source of complexity in the problem.
	This is evidenced by our hardness results and also by the fact that previous algorithms have relied on restrictions on such holes (or alternatively on the allowed partition).
	Furthermore, assembly sequencing specifically for assemblies with holes was recently raised as an open problem already for unit-grid squares~\cite{DBLP:conf/soda/Balanza-Martinez20}. %

\section{\newsatCap}
	\label{sec:sat}

	In this section we introduce a restricted version of \textsc{Monotone Planar 3-SAT}~\cite{DBLP:conf/cocoon/BergK10} and prove that even this restricted version is NP-complete. We begin with the definition of \textsc{Monotone Planar 3-SAT}.

	\paragraph{\textsc{Monotone planar 3-SAT}} Let $\phi = \bigwedge C_i$ be a \textsc{3-SAT} formula over $n$ variables, where each clause 
	$C_i$ is the disjunction of at most three\footnote{We remark that \textsc{3-SAT} formulas are sometimes defined as having \emph{exactly} three literals in each clause. We allow clauses with two literals since the restricted 3-SAT versions considered here use such clauses in their hardness reductions.} literals, each of which is either a variable or its negation.
	We call a clause with $k$ literals a $k$-clause.
	We consider the bipartite graph $G_\phi$ that contains a vertex for each variable and for each clause, and it 
	has an edge between a variable vertex
	and a clause vertex if and only if the variable appears in the clause.
	Lichtenstein \cite{DBLP:journals/siamcomp/Lichtenstein82} introduced the \textsc{Planar 3-SAT} problem, which requires $G_\phi$ to be planar, and showed that it is NP-complete. Furthermore, Knuth and Raghunatan \cite{DBLP:journals/siamdm/KnuthR92} have shown that the graph $G_\phi$ of a  \textsc{Planar 3-SAT} instance can be drawn in the so-called \textit{rectilinear} embedding.
	In this embedding, all vertices can be drawn as unit-height rectangles, with all the variable-vertex rectangles centered on a fixed horizontal strip called the \textit{variable row} and every edge is a vertical line segment, which does not cross any rectangles, as shown in Figure~\ref{fig:mpsat_example}. We denote the variables by $x_1, \ldots, x_n$, according to their left to right order on the variable row. \textsc{Planar 3-SAT} remains NP-complete when $G_\phi$ is given as a rectilinear embedding.
	
	\begin{figure}[htb]
		\centering
		\includegraphics[width=0.42\linewidth]{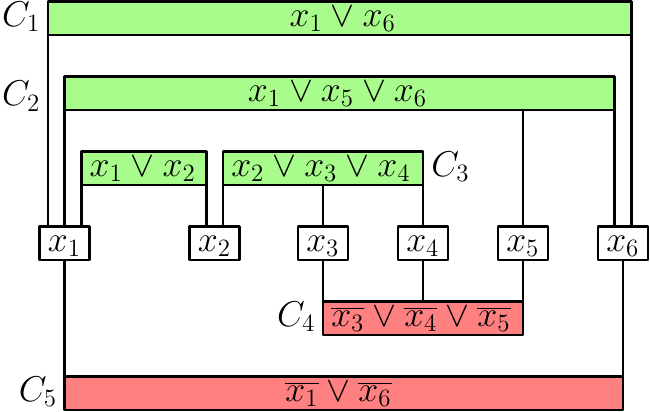}
		\caption{\textsc{Monotone Planar 3-SAT} instance in a rectilinear embedding, which is also an instance of \newsatshort{}.}
		\label{fig:mpsat_example}
	\end{figure} 
	
	A clause is \textit{positive (resp. negative)} if it contains only positive (resp. negative) literals. An instance of \textsc{3-SAT} is \textit{monotone} if it only has positive or negative clauses. \textsc{Monotone Planar 3-SAT} is a restriction of  \textsc{Planar 3-SAT} to monotone instances in which the rectilinear embedding of $G_\phi$ has all positive clause-vertices above the variable row and all negative clause-vertices below it, as illustrated in Figure~\ref{fig:mpsat_example}.
	\textsc{Monotone Planar 3-SAT} is also known to be NP-complete~\cite{DBLP:conf/cocoon/BergK10}.

	Our new \textsc{3-SAT} version is a further restriction of \textsc{Monotone Planar 3-SAT} in which each 3-clause has (at least) two variables that are consecutive on the variable row:
	
	\paragraph{\textsc{Monotone Planar 3-SAT \newrestriction{} (\newsatshort{})}}
	An instance is a \textsc{Monotone Planar 3-SAT} formula $\phi$, where any clause with 3 variables is of the form $(x_j \lor x_k \lor x_{k+1})$ or $(\lnot x_j \lor \lnot x_k \lor \lnot x_{k + 1})$.

	We call the two consecutive variables of a 3-clause $C$ the \textit{\nvarstext{}} of $C$, which we denote by $\nvars{C}$, i.e., $\nvars{(x_j \lor x_k \lor x_{k+1})} = \{x_k, x_{k+1}\}$.
	Note that we do not require 2-clauses to have \nvarstext{}. See Figure~\ref{fig:mpsat_example} for an example instance.
	
	Before proving the hardness of \newsatshort{}, we specify additional assumptions on instances of \newsatshort{} that are useful for our reduction from this variant in Section~\ref{sec:hardness-assembly}.
	To this end, we give definitions that are based on a given rectilinear embedding of $G_\phi$, 
	which capture the hierarchical relationship between clauses.
	
	\paragraph{Enclosing/Child/Parent/Root Clauses}
	Let us fix a rectilinear embedding of $G_\phi$ and let $C$ and $C'$ be two clauses that are on the same side of the variable row.
	We say that $C$ \emph{encloses} $C'$ if one can draw a vertical line segment $s$ connecting $C$ and $C'$, and also $C$ is vertically further from the variable row than $C'$.
	$C$ is called the \emph{parent} of $C'$ if it encloses $C'$ and furthermore we can draw $s$ without crossing any clauses.
	In this case, we say that $C'$ is the \emph{child} of $C$.
	A clause that does not have a parent is called a \emph{root} clause.
	For example, in Figure~\ref{fig:mpsat_example} $C_1$ encloses $C_3$, but is only the parent of $C_2$.
	
	We now highlight the additional assumptions we make for a given instance of \newsatshort{}:
	
	\begin{enumerate}[label=A\arabic*]
		\item
		On each side of the variable row, there is exactly one root clause $C_r$ such that $C_r$ contains two literals and encloses all other clauses on that side.
		\label{sat-ass:root}
		
		\item
		For each clause with three variable $C$, where $\nvars{C} = \{x_i, x_{i+1}\}$, no child clause of $C$ appears horizontally between the edges $(C, x_{i}), (C, x_{i+1})$ in $G_\phi$.
		\label{sat-ass:nochild}
	\end{enumerate}
	
	The assumptions hold for the instance in Figure~\ref{fig:mpsat_example}; $C_1$ and $C_5$ are each the unique root 2-clause on their respective side.
	The following straightforward lemma shows that we may assume \ref{sat-ass:root} and \ref{sat-ass:nochild} without any loss of generality.

	\begin{lemma} \label{lemma:root_clause}
		Let $\phi \in$ \newsatshort{}.
		Then $\phi$ can be converted in polynomial time to an equivalent Boolean formula $\phi' \in$~\newsatshort{} that has a rectilinear embedding where~\ref{sat-ass:root} and~\ref{sat-ass:nochild} hold.
	\end{lemma}
	
	\begin{proof}
		To have~\ref{sat-ass:root} we modify $\phi$ as follows:
		Introduce a variable at each end of the variable row, and add a 2-clause $C_r$ containing the new variables such that it encloses the rest of the positive clauses.
		Any existing positive root clause must now be the child of $C_r$.
		A unique negative root 2-clause containing the two new variables is obtained in a symmetric manner.
		The two new root clauses can always be satisfied, so $\phi'$ is equivalent to $\phi$.
		
		Now let $C$ be a clause with $\nvars{C} = \{x_i, x_{i+1}\}$ for which~\ref{sat-ass:nochild} does not hold.
		Then the child clause $C'$ violating~\ref{sat-ass:nochild} must contain only the variables $\{x_i, x_{i+1}\}$.
		However, since $\phi$ is monotone, $C'$ makes $C$ redundant and so we can remove $C$ from $\phi$.
	\end{proof}

	\begin{theorem}
		\textsc{Monotone Planar 3-SAT with Neighboring Variable Pairs} is NP-complete.
		\label{thm:sat}
	\end{theorem}
	\begin{proof}
		The problem is clearly in NP. We prove NP-hardness by a reduction from \textsc{Monotone Planar 3-SAT}. Let $\phi$ be a \textsc{Monotone Planar 3-SAT} instance given as a rectilinear embedding. We say that a 3-clause is \textit{valid} if it has a neighboring variable pair.
		The reduction works in two stages.
		In the first stage, we modify each 3-clause in $\phi$ so that it is valid.
		This stage introduces clauses that are not monotone, which are fixed, i.e., transformed to be monotone, in the second stage.
		
		\begin{figure}[htb]
			\centering
			\centering
			\includegraphics[width=0.55\textwidth]{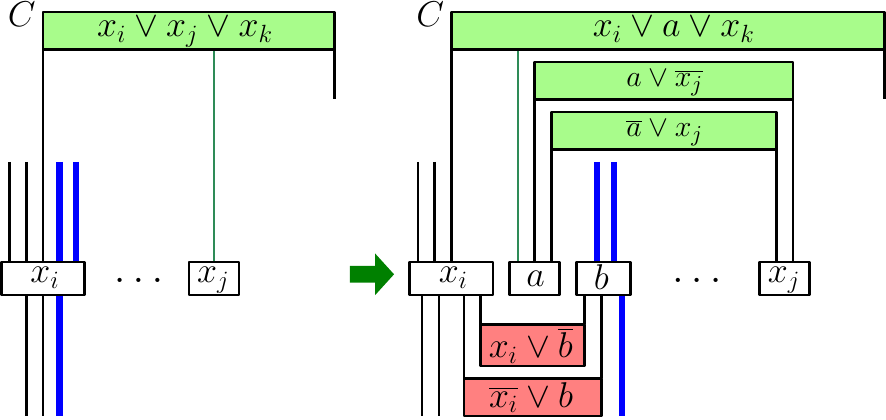} 
			\caption{Modifying a 3-clause so that it has a neighboring variable pair.
				The ellipsis represents the location of the variables $x_{i+1}, \ldots, x_{j-1}$.
			}
			\label{fig:PSAT_R1}
		\end{figure}
		
		\paragraph{Stage 1: making 3-clauses valid}
		Let $C=\clause{x_i}{x_j}{x_k}$, $i < j < k$, be a positive 3-clause, which therefore lies above the variable row.
		We modify it as follows; refer to Figure~\ref{fig:PSAT_R1} throughout. (Negative clauses, which lie below the variable row, are modified symmetrically.) Introduce two variables $a$ and then $b$ immediately to the right of $x_i$ and add the following \textit{equivalence clauses} that enforce $a=x_j$ and $b=x_i$:
		Add $(x_i \lor \overline{b})$ and $(\overline{x_i} \lor b)$ below the variable row so that they do not enclose any existing clauses.
		Add $(a \lor \overline{x_j})$ and $(\overline{a} \lor x_j)$ above the variable row so that they enclose all the existing clauses located horizontally between edges $(C,x_i)$ and $(C,x_j)$.
		
		We then apply the transformation \rewire{C, x_i, b}, which replaces $x_i$ by $b$ in each clause $C'$ that is either (i) above the variable row and has an edge $(C', x_i)$ that is to the right of the edge $(C, x_i)$ or (ii) below the variable row and has \nvars{C'}=$\{x_i, x_{i+1}\}$.
		Finally, we replace $x_j$ by $a$ in clause $C$, so that now $C= \clause{x_i}{a}{x_k}$, which makes $C$ valid since $x_i$ and $a$ are neighbors.
		
		The modification preserves the rectilinear embedding as we can shrink existing variables and clauses to create room for the new variables and clauses. The thick blue edges 
		in Figure~\ref{fig:PSAT_R1} are shifted right from $x_i$ to $b$ due to \rewire{C, x_i, b}. Any clause that was connected to $x_i$, is now connected to $b$ instead.
		The modified formula is equivalent to $\phi$ because the equivalence clauses are satisfied if and only if $a=x_j$ and $b=x_i$ and the variables are replaced accordingly. We verify that the modification does not invalidate a 3-clause already having \nvarstext{}: Observe that only a 3-clause $C'$ with the \nvarstext{} $\{x_i, x_{i+1}\}$ in $\phi$ may potentially be invalidated. However, by applying \rewire{C, x_i, b} we would have $\nvars{C'}=\{b, x_{i+1}\}$. Finally, all the clauses except the four new equivalence clauses remain monotone and on the correct side of the variable row.
		
		\begin{figure}[htb]
			\centering
			\includegraphics[width=0.42\textwidth]{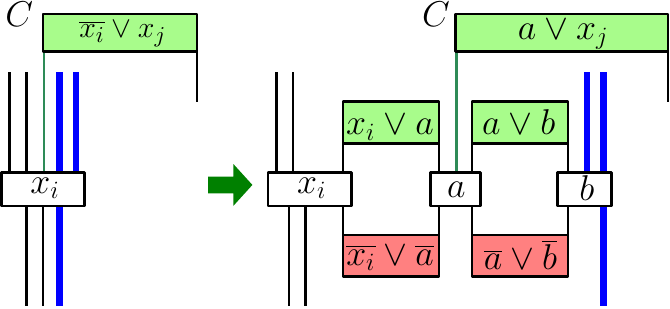} %
			\caption{Fixing a non-monotone 2-clause.}
			\label{fig:sat2}
		\end{figure}
		
		\paragraph{Stage 2: fixing non-monotone clauses}
		By applying the modification above to each 3-clause in $\phi$, we transform it to an equivalent formula in which each 3-clause has a neighboring variable pair. However, we now have to address the new non-monotone 2-clauses. We therefore apply a transformation nearly identical to the one used in the NP-completeness proof of \textsc{Monotone Planar 3-SAT}~\cite{DBLP:conf/cocoon/BergK10}, which fixes one non-monotone 2-clause at a time.
		Assume that the variables have been relabeled according to their order on the variable row, and let $C=\clausetwo{\overline{x_i}}{x_j}$, $i < j$ be a non-monotone clause. We make $C$ positive if it is above the variable row and negative if it is below it, as per the requirements of \textsc{Monotone Planar 3-SAT}. We assume the former case and modify $C$ as follows (the latter case, as well as the case of $j < i$, are handled similarly).
		Introduce two variables $a$ and $b$ immediately to the right of $x_i$ and add the following equivalence clauses, which are now monotone, that enforce $x_i=\overline{a}=b$; see Figure~\ref{fig:sat2}:
		\begin{gather*}
			(x_i \lor a) \land (\overline{x_i} \lor \overline{a}) \land (a \lor b) \land (\overline{a} \lor \overline{b}).
		\end{gather*}
		The new clauses do not enclose existing clauses.
		As in the first stage, we then apply \rewire{C, x_i, b}, which replaces $x_i$ by $b$ as defined above.
		Then, we replace $x_i$ by $a$ in clause $C$, which turns $C$ into a positive clause, as required.
		
		Arguing as in the first stage of the reduction, the modified formula is equivalent to $\phi$ and the modified rectilinear embedding is valid. As before, performing \rewire{C, x_i, b} guarantees that the modification preserves the neighboring variable pair of each 3-clause. Lastly, the new clauses, as well as $C$, are monotone and on the correct side of the variable row.
		
		By applying the modification to each non-monotone 2-clause, we transform $\phi$ to an equivalent formula that meets the requirements: it is monotone, all the positive (resp. negative) clauses are above (resp. below) the variable row, and each 3-clause is valid. Finally, it is easy to verify that the reduction can be done in polynomial time.
	\end{proof}

\section{Hardness of Connected Assembly Partitioning}
	\label{sec:hardness-assembly}
	In this section, we prove the NP-completeness of the \textsc{PPCST} problem
	even when the assembly parts are unit-grid squares, i.e.,
	$A \subset \{(i,j)+ [-0.5,0.5]^2 \mid i,j\in \dZ\}$, we call it a \textit{\gridass{}}; Figure~\ref{fig:intro} shows two such assemblies.
	We refer to the problem variant where the input assembly $A$ is a \gridass{} as \textsc{PPCST-GRID}.
	
	We note that \textsc{PPCST} in NP, since we can verify whether a given subassembly corresponds to a valid connected partition in polynomial time (Section~\ref{sec:alg-prelim} gives the required details).
	
	\subsection{The Polygonal Construction}\label{ssec:polygon-construction}\hfill\\
	Given $\phi \in$ \newsatshort{}, we construct a planar assembly $A \coloneqq A(\phi)$ that has a connected partition in \dirup{}
	if and only if $\phi$ is satisfiable.
	We call such a partition of $A$ a \emph{valid connected partition} and denote it by the subassembly $\up{} \subset A$, which contains the parts that are translated up.
	We denote the complement subassembly $A \setminus \up$ by \down{}.
	We construct the assembly $A$ based on the rectilinear embedding of $G_\phi$.
	The variables are represented by fixed-height rectangles whose bottom edges lie on a common horizontal line.
	The left-to-right order of the variable rectangles corresponds to the order of variables on the variable row in $G_\phi$. %
	With a slight abuse of notation, we denote both the variable and the corresponding variable rectangle by $x_i$.
	The main idea is the following correspondence between a partition ${\up{} \subset A}$ and an assignment to $\phi$: $x_i$ is assigned \true{} if and only if $x_i \in \up{}$ for the corresponding variable rectangle.
	We denote such an assignment by \ass{}.

	\begin{figure}[htb]
		\centering\includegraphics[width=0.39\textwidth]{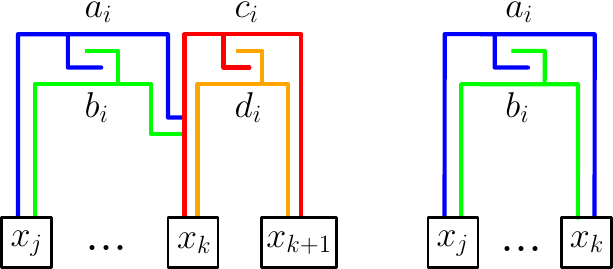}
		\caption{(Left) A clause gadget for $(x_j \vee x_k \vee x_{k+1})$ together with the three variable rectangles it touches. (Right)~A clause gadget for $(x_j \vee x_k)$.
		}
		\label{fig:clause}
	\end{figure}
	
	\paragraph{The clause gadget}
	
	We describe the clause gadget, focusing first on positive 3-clauses.
	Figure~\ref{fig:clause} (left) shows the gadget for a clause $C_i = (x_j \vee x_k \vee x_{k+1})$.
	The gadget is located above the variable row in $A$. %
	It is composed of four parts $\clausea{i}, \clauseb{i}, \clausec{i}, \claused{i}$, which are rectilinear polygons.
	In all the figures, $a_i$ parts are blue, $b_i$ parts are green, $c_i$ parts are red, and $d_i$ parts are orange.
	We use colored line segments to represent very thin rectangular portions of a polygon.
	Each 3-clause $C_i$ has $x_k, x_{k+1}$ as its neighboring variables, as required by \newsatshort{}.
	We always have both $c_i$ and $d_i$ touch $x_k$ and $x_{k+1}$, as shown.
	We remark that when combining all the clause gadgets in the complete construction, we will slightly modify parts $a_i$ and $b_i$. This modification will not change the functionality of the gadget, which we describe next.

	\begin{figure}[!b]
		\centering
		\begin{tabular}{ccc}
        \begin{adjustbox}{valign=c}
		\includegraphics[width=0.43\textwidth]{figures/MPSAT.pdf}
                \end{adjustbox}
			&\hspace*{0.2in}&
		\begin{adjustbox}{valign=c}
        \includegraphics[width=0.45\textwidth]
        {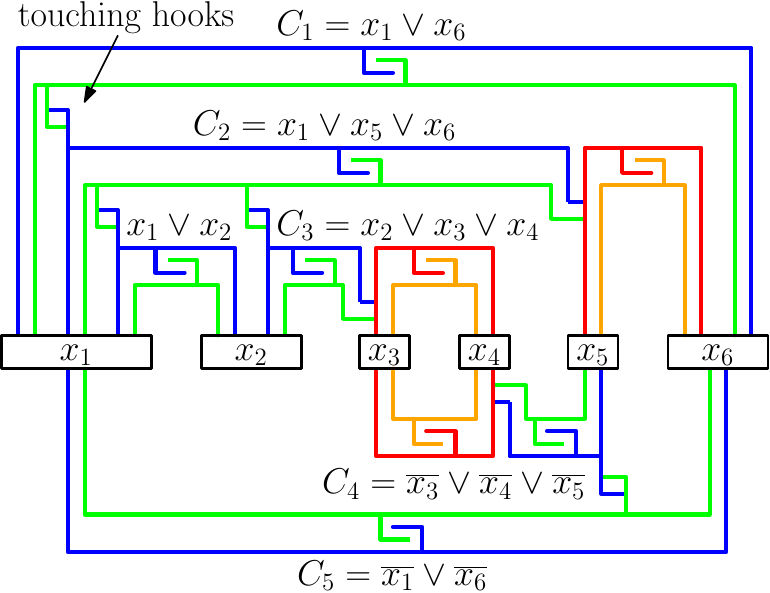} 
        \end{adjustbox}
			\\
			(a)&&(b)
		\end{tabular}
		\caption{(a) An instance of \newsatshort{}, in a rectilinear embedding; 
			(b)  the corresponding assembly construction $A$.
		}
		\label{fig:full}
	\end{figure}

	\begin{figure}[!b] 
		\centering
		\begin{minipage}{0.45\textwidth}
			\centering
			\includegraphics[width=0.96\textwidth]{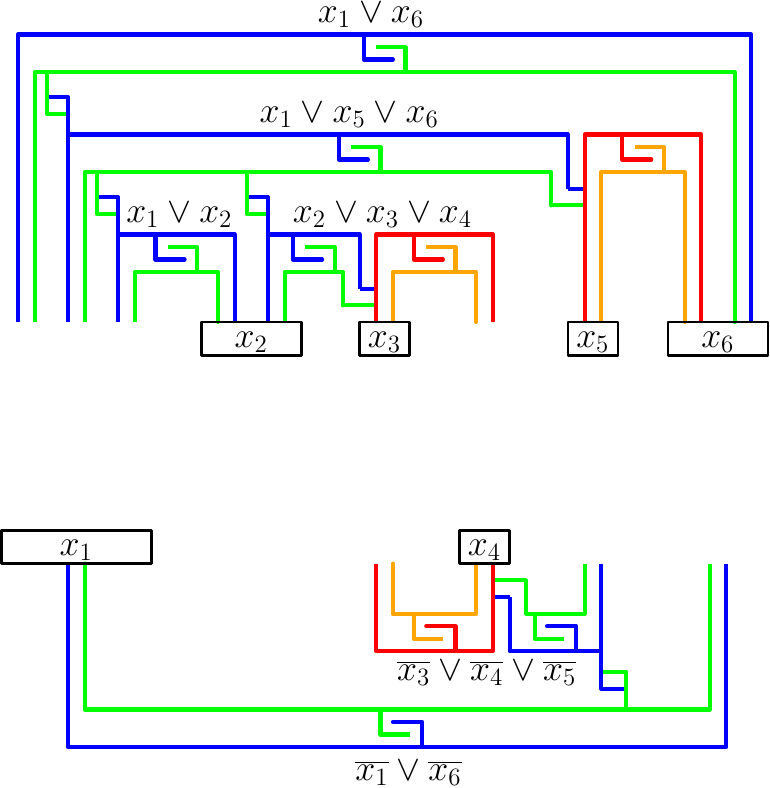} 
			
		\end{minipage}
		\hfill
		\begin{minipage}{0.46\textwidth} %
			\centering
			\includegraphics[width=0.96\textwidth]{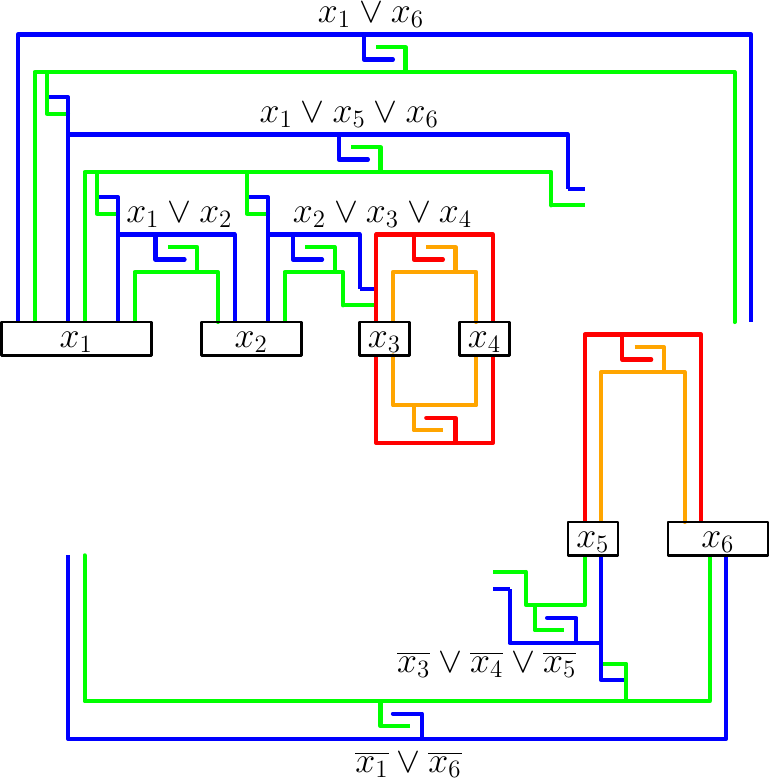} 
		\end{minipage}
		\centering
		\caption{(left) Partition for the assignment $x_1=F,x_2=T,x_3=T,x_4=F,x_5=T,x_6=T$. %
			(right) Partition for the assignment $x_1=T,x_2=T,x_3=T,x_4=T,x_5=F,x_6=F$.
		}
		\label{fig:polygonal_partition_examples}
	\end{figure}
	
	Parts $p, p'\in A$ are said to \textit{interlock} if they are not separable by a vertical translation, i.e., $p \in \up{} \Leftrightarrow p' \in \up{}$.
	For example, parts \clausea{i} and \clauseb{i} of the gadget interlock, as do parts \clausec{i} and \claused{i}.
	Let us assume that $\clausea{i} \in \up{}$ (we later force this to always be the case).
	Due to the interlock, we must have $\clauseb{i} \in \up{}$, and since \clausea{i} and \clauseb{i} are both in \up{}, there needs to be a connection between them in order for \up{} to be connected (i.e., there must exist a path between \clausea{i} and \clauseb{i} in the adjacency graph $\mathcal{G}(\up{})$). The only way to connect them is for $x_j$ or $\clausec{i}$ to be in \up{}. In this sense, \clausea{i} and \clauseb{i} act analogously to a logical OR gate with the inputs $x_j \in \up{}, \clausec{i} \in \up{}$.
	Therefore, if $x_j \in \down{}$, we must have $\clausec{i} \in \up{}$.
	In that case, \clausec{i} and \claused{i} act in an analogous manner to \clausea{i} and \clauseb{i}.
	Namely, we must have $\claused{i} \in \up{}$, and in order to have a connection between \clausec{i} and \claused{i}, we must have either $x_k \in \up{}$ or $x_{k+1} \in \up{}$.
	To summarize, if $a_i \in \up{}$ then the gadget guarantees that at least one of the variable rectangles of $C_i$ (i.e., $x_j, x_k, x_{k+1}$) is also in \up{}.
	We therefore get that any partition with $a_i \in \up{}$ corresponds to an assignment satisfying $C_i$:
	
	\begin{lemma}
		For a positive clause $C_i$, if $\clausea{i} \in \up{}$ then \ass{} satisfies $C_i$.
		\label{lemma:clause_gadget}
	\end{lemma}

	We now discuss the other types of clauses.
	If $C_i$ is a positive 2-clause $(x_j \vee x_k)$, the corresponding gadget contains only the parts \clausea{i} and \clauseb{i} modified so that their right sides touch $x_k$ instead of \clausec{i}, as shown in Figure~\ref{fig:clause} (right).
	Lemma~\ref{lemma:clause_gadget} holds for 2-clauses as well using a similar argument.
	As for negative clause gadgets, they are drawn the same way, except that they are mirrored through the variable row (as they are placed below it), as will shortly be shown in the example of the whole construction in Figure~\ref{fig:full} (right).
	In general, all claims for positive clauses hold symmetrically for negative clauses, which we do not always reiterate.
	
	\paragraph{The complete construction} %
	We now incorporate all the clause gadgets into the complete construction of $A$.
	All the clause gadgets are placed according to each clause's position in the rectilinear embedding of $G_\phi$, which allows each clause gadget to touch the appropriate variable rectangles;
	see  Figure~\ref{fig:full}.
	
	We now introduce a slight change to the clause gadgets, based on the hierarchical relationship between clauses, as defined in Section~\ref{sec:sat} after introducing \newsatshort{}.
	The change is the addition of \emph{touching hooks} to each parent-child pair of clauses, which are similar to the hooks used within a single clause gadget (to interlock $a_i$ and $b_i$), except that they are now touching.
	More precisely, for each parent-child pair of positive clauses $(C_i,C_j)$, we add a hook emanating downwards from $a_i$ and another one emanating upward from $b_j$ such that the parts $a_i$ and $b_j$ interlock and touch (negative clauses are handled symmetrically); see Figure~\ref{fig:full}.
	
	We now use the fact that on each side of the variable row there is a unique root clause, which contains two literals, and encloses all other clauses on that side (see assumption~\ref{sat-ass:root} and Lemma~\ref{lemma:root_clause}).
	Let $C_r$ denote the unique positive root clause.
	The following lemma states that in order to translate \clausea{r} (the top part of the positive root clause gadget) upward, we must also translate all the $\set{a_i}$ parts of positive clause gadgets.
	
	\begin{lemma} \label{lemma:interlocking_a_parts}
		For a positive non-root clause $C_i$, $a_i$ interlocks with \clausea{r}.
	\end{lemma}

	\begin{proof}
		Let $C_i$ be a positive non-root clause and let $C_j$ be its parent clause.
		Then $a_i$ interlocks via the touching hooks with part $b_j$ of its parent clause.
		Since $b_j$ interlocks with $a_j$, we get that $a_i$ interlocks with $a_j$.
		It is straightforward to inductively conclude that $a_i$ transitively interlocks with \clausea{r}.
	\end{proof}
	
	Hence, all the positive $\{\clausea{i}\}$ (and in fact $\{\clauseb{i}\}$) parts must belong to the same set of the partition.
	Note that this is not the case for $\{\clausec{i}\}$ and $\{\claused{i}\}$ parts, which do not interlock with \clausea{r}, resulting in some flexibility that we will exploit.

	We now prove the correctness of the construction by showing that $\phi$ has a satisfying assignment if and only if the assembly $A$ has a valid connected partition.
	\begin{lemma}
		Let $\up{}$ be a valid connected partition of $A$.
		Then \ass{} satisfies $\phi$.
	\end{lemma}
	\begin{proof}
		Assume that \up{} is a connected partition of $A$, namely that all parts in \up{} can be translated upward without colliding with \down{} and that each of \up{} and \down{} is connected.
		We need to show that the corresponding assignment \ass{}, in which $x_i$ is \true{} if and only if $x_i \in \up{}$, satisfies $\phi$.
		We show that each positive clause in $\phi$ is satisfied, which will symmetrically give the same results for each negative clause.
		First, any part of $A$ that translates up collides with $a_r$.
		This is true because the root clause $C_r$ encloses all other positive clauses and since we may assume that each variable appears in a positive clause.
		Hence, we have $\clausea{r} \in \up$.
		By Lemma~\ref{lemma:interlocking_a_parts}, we also have $\clausea{i} \in \up{}$ for each positive clause $C_i$.
		By Lemma~\ref{lemma:clause_gadget}, each $C_i$ must then be satisfied, which concludes the proof.
	\end{proof}
	
	For the other direction, we assume that $\phi$ has a satisfying assignment $\mathcal{A}$ and specify a corresponding partition \up{}.
	Of the variable rectangles, \up{} includes exactly those that correspond to variables assigned \true{}.
	For each positive clause $C_i$, we include \clausea{i} and \clauseb{i} in \up{}.
	If $C_i$ is a 3-clause, we determine where \clausec{i} and \claused{i} go based on the assignment $\mathcal{A}$ as follows:
	Assign \clausec{i} and \claused{i} to \up{} if and only if $\mathcal{A}$ assigns \emph{at least one} of $\nvars{C_i}$ to be \true{}.
	If $C_i$ is a negative 3-clause, assign \clausec{i} and \claused{i} to \up{} if and only if $\mathcal{A}$ assigns \emph{both} $\nvars{C_i}$ to be \true{}.
	
	See Figure~\ref{fig:polygonal_partition_examples} for examples of partitions of the construction in Figure~\ref{fig:full}.
	We now prove the connectivity and separability properties of the partition \up{}.

	\begin{lemma}
		\up{} and \down{}, as defined based on a satisfying assignment $\mathcal{A}$, are separable by an infinite vertical translation.
		\label{lemma:no_collision}
	\end{lemma}
	\begin{proof}
		If all the parts above the variable row translate up and all the parts below the variable row do not, then it is easy to verify that no collisions can occur.
		Therefore, we assume that this is not the case.
		More precisely, assume without any loss of generality that some part of $A$ above the variable row translates down.
		This part must be some $d_i$ for a positive clause $C_i$ (since $c_i$ and $d_i$ always translate in the same direction we may treat them as one part; hence we refer only to $d_i$).
		Using the neighboring variable pairs property of \newsatshort{}, $d_i$ touches neighboring variable rectangles, and so there are no other parts below $d_i$ that are above the variable row (see assumption~\ref{sat-ass:nochild} in Section~\ref{sec:sat}).
		Therefore, $d_i$ may only collide with a part that is initially below the variable row.
		Let us assume, for a contradiction, that $d_i$ collides with some $d_j \in \up$, where $C_j$ is a negative clause (the same arguments hold for a collision with $c_j$).
		Since $d_i$ and $d_j$ both touch some neighboring variable rectangles in $A$, there must be a variable rectangle $x_k$ that they both initially touch, in order to collide.
		By the definition of the partition, $\nvars{C_i} \subset \down{}$, so we have $x_k \in \down$.
		However, if $x_k \in \down$ then $d_j$ would be assigned to $\down$ instead of $\up$ as we assumed, which is a contradiction.
	\end{proof}
	 
	\begin{lemma}
		\up{} and \down{}, as defined based on a satisfying assignment $\mathcal{A}$, are both connected subassemblies.
		\label{lemma:conn_partition}
	\end{lemma}
	
	\begin{proof}
		We verify that \up{} is connected by considering each type of part in it (the symmetric logic applies to \down{}).
		For each positive clause $C_i= (x_j \vee x_k \vee x_{k+1})$, let $S_i \subset \up{}$ denote the parts of the respective clause gadget and the variable rectangles $x_j,x_k,x_{k+1}$ that translate up.
		We establish that $S_i$ is connected (the case of a 2-clause, with $S_i$ defined analogously, is simpler):
		We have $a_i, b_i \in S_i$.
		Given that $\mathcal{A}$ satisfies $C_i$, it is easy to verify that $a_i$ and $b_i$ are connected either via $c_i$ or via $x_j$, which $a_i$ and $b_i$ touch.
		If $c_i$ and $d_i$ are in $S_i$, then they are connected via either $x_k$ or $x_{k+1}$, one of which is in $S_i$ by the definition of the partition.
		
		For a child-parent pair of clauses, $C_i, C_j$ the touching hooks (see Figure~\ref{fig:full}) further imply that $S_i$ is connected to $S_j$.
		As a result, $S_i$ must be connected to $S_r$, where $C_r$ is the unique positive root clause.
		Note that any variable rectangle in $\up{}$ must touch some clause gadget part in $\up{}$, since we may assume that the corresponding variable appears in a positive clause.
		Therefore, we obtain that the subassembly $S \subseteq \up$ containing positive clause gadget parts and variable rectangles is connected.
		
		It remains to examine some $c_i \in \up{}$ where $C_i$ is a negative clause (the case of $d_i$ is the same).
		Since $c_i \in \up{}$, then by definition we have $\nvars{C_i} \subset \up{}$.
		In other words, $c_i$ is connected to $S$ via the neighboring variables it touches.
		Therefore, all the parts in \up{} form a connected subassembly.
	\end{proof}
	
	In summary, given a satisfying assignment to $\phi$, we have defined a valid connected partition $\up{} \subset A$, thus concluding the argument in the remaining direction.
	Overall, we obtain the following result.
	
	\begin{theorem}
		\textsc{PPCST} is NP-complete.
		\label{thm:ppcst_hardness}
	\end{theorem}

	\subsection{The Unit-Grid Square Construction}
	\label{subsec:square-construction}\hfill\\
	
	We now convert $A$ to a \gridass{} $A_{sq}$ to show that the restricted problem variant \textsc{PPCST-GRID} remains NP-complete.
	We show how to convert each polygon $p \in A$ to a \textit{quasi-polygon} --- an equivalent collection of squares that emulates $p$.
	
	Figure~\ref{fig:full_squares} shows the converted version of the polygonal construction in Figure~\ref{fig:full}.
	In this figure, each line segment signifies a contiguous part of a single grid row or column occupied by squares.
	Each variable now corresponds to a variable gadget, which we temporarily consider to be a rigid part (like the variable rectangles in the polygonal construction) for ease of exposition.
	The full details of the variable gadgets are given at the end of this section.
	The definition of the assignment \ass{} now refers to the variable gadgets, i.e., $x_i \in \phi$ is assigned \true{} if and only if the corresponding variable gadget is in $\up{}$.
	
	Figure~\ref{fig:clause_conversion} shows a "zoomed in" view of a converted positive 3-clause, where individual squares are shown. It shows the four quasi-polygons $a_i, b_i, c_i, d_i$ that correspond to the polygonal parts of the same name, respectively.
	Each quasi-polygon maintains the same color as before (blue, green, red, and orange for $a_i,\ldots, d_i$, respectively).
	
	\begin{figure}[htb]
		\centering
		\includegraphics[width=0.87\textwidth]{figures/full-squares_twisted_path_w_zoom.pdf}
		\caption{A complete grid square assembly construction (right). The shaded rectangle is shown zoomed in on the left, where $C(b_i)$ is drawn using thicker pink line segments.}
		\label{fig:full_squares}
	\end{figure}
	
	We now describe the conversion of the polygonal construction in more detail.
	For brevity, the description focuses only on the new elements and intricate details.
	At the top of the assembly $A_{sq}$, we add a \textit{blocking} quasi-polygon $B$ (drawn in black in Figure~\ref{fig:full_squares}).
	The top row of $B$ is wide enough so that it blocks all the other squares in the construction from above.
	A symmetric quasi-polygon $B'$ that blocks all other squares from below is added at the bottom of the construction.
	That is, \up{} (resp. \down{}) must contain a square of $B$ (resp. $B'$).
	Therefore, we denote by $t$ some square in $B' \cap \down{}$, which is non-empty.
	
	The main intricacy that arises in quasi-polygons is that they may get split up by the partition instead of remaining a single rigid piece.
	We therefore introduce a lemma that will be used to show that such splits cannot occur in critical places.
	
	\begin{lemma} \label{lemma:split}
		Let $s \in A_{sq}$ be a square and $S \subset A_{sq}$ be a subassembly such that $S$ is a connected component of $\mathcal{G}(A_{sq} \setminus \{s\})$ and $t \notin S$.
		Then, $s \in \up{} \Rightarrow S \subset \up{}$.
	\end{lemma}
	\begin{proof}
		Assume that $s \in \up{}$ and also assume for contradiction that there exists some square $s' \in S \cap \down{}$.
		Since \down{} is connected, it must contain a path between $s'$ and $t$ that $s$ cannot be a part of, since $s \in \up{}$.
		This path is contained in $A_{sq} \setminus \{s\}$, which implies that $s'$ and $t$ must be in the same connected component of $\agraph(A_{sq} \setminus \{s\})$.
		Therefore, we have $t \in S$, which is a contradiction.
	\end{proof}

	\begin{lemma} \label{lemma:blocker}
		For any connected partition $\up{} \subset A_{sq}$ we have $B \subset \up{}$.
	\end{lemma}
	\begin{proof}
		In any partition, at least one square $s$ in the top row of $B$ must move up, otherwise no square can move up.
		Let $S$ and $S'$ be the two connected components of $\agraph(B \setminus \{s\})$, where $S$ corresponds to the subassembly containing the rightmost squares of $B$.
		By applying Lemma~\ref{lemma:split} to $S$, we have $S \subset \up{}$.
		As a result, the squares of $a_r$ that are above $S$ (represented by a portion of the topmost blue line segment in Figure~\ref{fig:full_squares}) must also move up.
		For \up{} to be connected, it must contain a path $P$ between these squares of $a_r$ and $S$.
		$P$ must contain $S'$, so overall we have $B \subset \up$.
	\end{proof}

	\begin{figure}[htb]
		\centering
		\includegraphics[width=0.39\textwidth]{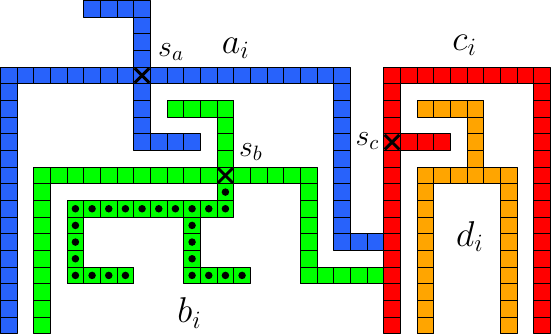}
		\caption{An approximate illustration of a clause gadget composed of unit-grid squares for a clause with two child clauses. We mark the key squares $s_a, s_b, s_c$ using a cross. $L(b_i)$ is the set of squares marked by dots.}
		\label{fig:clause_conversion}
	\end{figure}

	\paragraph{Converting clause gadgets}
	We turn our attention to clause gadgets, for which the main change is a modification of the so-called touching hooks:
	For each parent-child pair of clauses $(C_i, C_j)$, we add a gap between their two touching hooks, so that they no longer touch.
	Furthermore, all the hooks emanating downwards from $b_i$ now form a connected subassembly, denote by $L(b_i)$, which we call the \textit{lower hooks} of $b_i$.
	$L(b_i)$ connects to the rest of $b_i$ only via the square $s_b$; see Figure~\ref{fig:clause_conversion}.
	More precisely, $L(b_i)$ is the connected component of $\mathcal{G}(b_i \setminus \{s_b\})$ that neighbors the lower edge of $s_b$.
	We define $U(b_i)$ (\textit{upper hook} of $b_i$) to be the upper connected component of $\mathcal{G}(b_i \setminus \{s_b\})$, i.e., the component that neighbors the upper edge of $s_b$.
	(Recall that squares are considered touching only if they share an edge, so $\mathcal{G}(b_i \setminus \{s_b\})$ has four connected components.)
	We similarly define $U(a_i)$ and $L(a_i)$, as the upper and lower component, respectively, of $\mathcal{G}(a_i \setminus \{s_a\})$.
	
	The connection that the touching hooks used to provide between a child and parent clause is now maintained using a separate subassembly from the hooks, denoted by $C(b_i)$.
	We consider $C(b_i)$ to be a part of $b_i$ (for the purpose of defining a partition) but draw it in pink due to its intricacy.
	$C(b_i)$ touches the $a_j$'s of $C_i$'s child clauses without touching the upper hooks of the $a_j$'s nor $L(b_i)$.
	This is achieved by carefully placing the squares composing $C(b_i)$ between the said lower and upper hooks when there are many child clauses.
	See the zoomed-in view in Figure~\ref{fig:full_squares} for an example of a $C(b_i)$ when $C_i$ has multiple child clauses.
	
	We now prove a lemma analogous to Lemma~\ref{lemma:clause_gadget} in the polygonal construction.
	The requirement of Lemma~\ref{lemma:clause_gadget} is that the \textit{polygon} $a_i$ translates up, but because the corresponding \textit{quasi-polygon} $a_i$ is made of squares, we use an adapted requirement, which is that $U(a_i) \subset \up$:
	
	\begin{lemma}
		\label{lemma:square_clause_gadget}
		For a positive clause $C_i$, if $U(a_i) \subset \up$ then (i) $L(b_i) \subset \up$ and (ii) \ass{} satisfies $C_i$.
	\end{lemma} %
	\begin{proof}
		Let $C_i = (x_j \vee x_k \vee x_{k+1}), j < k$ be a positive clause.
		If $U(a_i) \subset \up$ then \up{} must contain a path connecting $U(a_i)$ to $B$, since $B \subset \up{}$ (by Lemma~\ref{lemma:blocker}).
		Observe that any such path must include the square $s_a$, so we have $s_a \in \up{}$.
		By invoking Lemma~\ref{lemma:split} on $L(a_i)$ (as $S$ in the lemma) and $s_a$, we get that $L(a_i) \subset \up$.
		Consequently, $L(a_i)$ forces the subassembly of $U(b_i)$ that is in its way, to translate up.
		As before, we must now have $L(a_i)$ and the subassembly of $U(b_i)$ that blocks $L(a_i)$ be connected by some path $P \subset \up{}$.
		Observe that in any case $P$ must contain the square $s_b$.
		This fact, together with Lemma~\ref{lemma:split}, results in $L(b_i) \subset \up$, which concludes the proof for (i).
		$P$ can either go through the variable gadget $x_j$ or go through the leftmost column of squares of $c_i$ (similarly to the two choices for connecting the polygonal counterparts of $a_i$ and $b_i$).
		Therefore, if $P$ does not go through the variable gadget $x_j$, it must contain a part of the leftmost column of $c_i$, in which case the square $s_c$ would also translate up, as it is in the way.
		By invoking Lemma~\ref{lemma:split} on $s_c$, the connected subassembly neighboring the right edge of $s_c$ (as $S$ in the lemma) must translate up.
		This forces a subassembly of $d_i$ to translate up, which means that we must have a path in \up{} connecting it and the subassembly of $c_i$ translating up.
		Such a path must contain one of the variable gadgets $x_k$ or $x_{k+1}$.
		In summary, one of the variable gadgets that correspond to variables in $C_i$ must translate up.
		This proves (ii) by the definition of $\ass{}$.
		The proof for 2-clauses is similar.
	\end{proof}
	
	We now use Lemma~\ref{lemma:square_clause_gadget}
	similarly to how we use Lemma~\ref{lemma:clause_gadget} in the hardness proof of \textsc{PPCST}.
	
	\begin{lemma}
		$\phi$ has a satisfying assignment if and only if $A_{sq}$ has a valid connected partition.
	\end{lemma}
	
	\begin{proof}
		If $\phi$ has a satisfying assignment, then the corresponding polygonal assembly $A$ has a connected partition, as specified in Section~\ref{ssec:polygon-construction}.
		This partition easily applies to $A_{sq}$ by treating each quasi-polygon as a regular polygon that must move rigidly together.
		All the squares comprising a quasi-polygon are assigned together to either \up{} or \down{} based on the direction the quasi-polygon is supposed to move in (and the new quasi-polygon $B$ would be sent up).
		The partition is valid because the connectivity and blocking relations between the quasi-polygons are the same as for the respective polygons.
		
		In the other direction, when given a connected partition $\up{} \subset A_{sq}$, we also proceed as before and take the assignment \ass{}, which assigns a variable to be \true{} if and only if the corresponding variable gadget moves up.
		Let $C_i$ be a positive clause.
		By Lemma~\ref{lemma:square_clause_gadget}(ii), it suffices to show that $U(a_i) \subset \up$.
		We begin with $C_r$, the unique root clause.
		The proof of Lemma~\ref{lemma:blocker}, which establishes that $B \subset \up$, also obtains $U(a_r) \subset \up$.
		For a $C_i$ with a parent $C_j$, if $U(a_j) \subset \up$ then by Lemma~\ref{lemma:square_clause_gadget}(i) we have $L(b_j) \subset \up$, which in turn forces $U(a_i) \subset \up$ (i.e., the leftmost square of $U(a_i)$ is forced to translate up due to blocking relations, which consequently forces the rest of $U(a_i)$ up due to the connectivity of \up{}).
		Therefore, by induction we get $U(a_i) \subset \up$ for all $i$.
		We may use all our previous claims in a symmetric manner to show that negative clauses are satisfied as well.
		In conclusion, the assignment \ass{} satisfies $\phi$.
	\end{proof}

\begin{figure}[htb]
    \centering
    \includegraphics[width=0.18\textwidth]{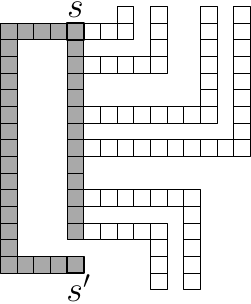}
    \caption{A variable gadget composed of unit-grid squares.}
    \label{fig:variable_gadget}
\end{figure}

	\paragraph{Converting variable rectangles} \label{sec:var gadget}
	We now describe how to convert variable rectangles to equivalent quasi-polygons.
	Figure~\ref{fig:variable_gadget} shows the quasi-polygon equivalent of a single variable rectangle.
	The gadget consists of the squares shaded gray in the figure.
	Let $V$ denote this set of squares.
	The unfilled squares are connections coming from clause gadgets.
	Each connection consists of a pair of paths, which touch the gadget from the right; in the figure there are two connections from positive clauses and one from a negative clause.
	We claim that replacing variable rectangles with this gadget preserves the correctness of the construction.
	To this end, we modify the definition of \ass{} so that it refers to the new gadget, i.e., a variable is assigned \true{} if and only if $V \subset \up{}$ (where $V$ is the corresponding gadget).
	To show that \ass{} still satisfies $\phi$ for a connected partition \up{}, we prove that all the squares in $V$ always move in the same direction, emulating a single rigid piece:
	\begin{lemma}
		For any connected partition $\up{} \subset A_{sq}$ we have either $V \subset \up{}$ or $V \subset \down{}$.
	\end{lemma}
	\begin{proof}
		We show that $V$ moves in the same direction as its top-right square $s$ (see Figure~\ref{fig:variable_gadget}):
		If $s \in \up{}$, then by Lemma~\ref{lemma:split} the connected subassembly of $V \setminus \{s\}$ resembling a C-shape must also move up. This subassembly includes the bottom-right square $s'$, which forces all the other squares of $V$ to move up too.
		Otherwise $s \in \down{}$, in which case the whole rightmost column of $V$ must also move down. Since \down{} is connected, all other squares in $V$ must then also move down to have a connection between $s$ and~$s'$.
	\end{proof}

	After converting the variable rectangles, $A_{sq}$ is a \gridass{} that has a valid connected partition if and only if $\phi$ is satisfiable.
	It is straightforward to verify that $A_{sq}$ fits inside an $O(m) \times O(n)$ rectangle, where $m$ and $n$ are the numbers of clauses and variables in $\phi$, respectively.
	We can construct $A_{sq}$ in polynomial time, which obtains the following result.
	
	\begin{theorem}
		\textsc{PPCST-GRID} is NP-complete.
	\end{theorem}

\section{Fixed-Parameter Tractable Algorithm}\label{sec:FPT}
	In this section, we show that \textsc{PPCST} is fixed-parameter tractable, with the size of the partition as the parameter. The following theorem states the main result:
	\begin{theorem}
		\label{theo:FPT}
		Let $A$ be a planar assembly consisting of $n$ polygonal parts with a total of $N$ vertices, 
		and let $\vec{d}$ be a direction in the plane.
		The assembly $A$ can be preprocessed in $O(N(n+\log N))$ time so that for an integer $k\le n$, 
		the existence of a connected partition of $A$ of size at most $k$ in direction $\vec{d}$
		can be decided in  $O(2^k n^2)$ time. 
		The algorithm also returns a connected partition of size at most $k$ if one exists.
	\end{theorem}
	
	The value of $k$ can be assumed to be at most $n/2$, for otherwise we can ask the ``mirrored'' version of the question, i.e., whether there exists a connected partition of size $n-k$ in direction~$-\vec{d}$.
	Furthermore, if $k$ is not given as part of the input and the goal is to find a connected partition of $A$, we can 
	run our algorithm and its ``mirrored'' version for all possible values of~$k$ from~$1$ to~$n/2$ until we find a 
	connected partition or the algorithm does not find a connected partition even for~$k=n/2$.
	If $k^*$ is the smallest size of a connected partition, the total running time of the algorithm, including preprocessing, is $O(N(n+\log N)+2^{k^*} n^2)$.
	
	\subsection{Preliminaries} \label{sec:alg-prelim}
	Let $A \coloneqq \set{P_1, \ldots, P_n}$ be a planar assembly, i.e., a set of polygons with pairwise-disjoint interiors, 
	let $N \coloneqq \sum_i |P_i|$ be the total number of polygon vertices, and let $\vec{d}$ be a direction. Without loss of generality, we assume that $\vec{d}$ is \dirup{}.
	We note that if a polygon $P_i$ has a hole that 
	contains another polygon $P_j$ then for any partition $S$, either both $P_i$ and $P_j$ are in $S$ or neither of them is in $S$. Therefore we can ignore a polygon $P_j$ if it lies in the hole of another polygon $P_i$ and fill all holes of 
	$P_i$. This preprocessing can be done in a total of $O(N\log N)$ time.
	Hereafter we assume that each part $P_i$ of $A$ is a simple polygon.
	
	We define two graphs on the polygons of $A$. For each $P_i$, we choose a point $v_i$ in the interior of $P_i$. Let $V=\{v_1,\ldots,v_n\}$. $V$ is the set of vertices in both graphs.
	For a subassembly $S \subseteq A$, we define $\nodef{S} \coloneqq \{v_j \mid P_j \in S\}$ and set 
	$\compl{S} \coloneqq A\setminus S$. 
	
	\paragraph{Adjacency graph} 
	The first graph is the \emph{adjacency graph} $\agraph(A)=(V,E)$ of $A$, defined in Section~\ref{sec:intro}.
	Recall that $\agraph(A)$ is an undirected graph with $(v_i,v_j)\in E$ if and only if $P_i$ and $P_j$ are 
	\emph{edge-connected}, i.e., 
	$P_i\cap P_j$ contains a line segment of non-zero length. Since polygons in $A$ have pairwise-disjoint interiors, $\agraph(A)$ is planar and has $O(n)$ edges.
	$\agraph(A)$ can be computed in $O(N\log N)$ time by using a simple sweep-line algorithm. 
	For a subset $S\subseteq A$, let $\agraph(S)$ denote the subgraph of $\agraph(A)$ induced by $\nodef{S}$. 
	We use $\agraph$ to denote $\agraph(A)$.
	For simplicity, with a slight abuse of notation, we say that $S$ is \emph{connected} if $\agraph(S)$ is connected, 
	and that $S_1, \ldots, S_r$ are connected components of $S$ if $\nodef{S_1}, \ldots, \nodef{S_r}$ are connected components of $\agraph(S)$.
	
	We construct the following planar embedding of $\agraph(A)$, which will be treated as a plane graph throughout.
	For each edge $(v_i, v_j) \in E$ we fix a point $q_{ij}$ in the interior of $\bd P_i\cap \bd P_j$ (such that $q_{ij} \notin \bd P_\ell$,
	for $\ell \notin \{i,j\}$). Set $Q_i =\{q_{ij} \mid (v_i,v_j) \in E\}$.  Consider the geodesic shortest path $\ph_{ij}$ from 
	$v_i$ within $P_i$  to each $q_{ij}\in Q_i$. If a geodesic $\ph_{ij}$ intersects the boundary of $P_i$ at a point other 
	than $q_{ij}$ or overlaps with another geodesic $\ph_{ij'}$, we perturb it slightly so that it lies in the interior of $P_i$ (except at $q_{ij}$) and does not intersect any  other geodesic except at their common source point $v_i$.
	We omit the straightforward details of this perturbation and, slightly abusing the notation, use $\ph_{ij}$ to denote the resulting perturbed path.
	For each edge $(v_i,v_j)\in E$, we view the path $\ph_{ij}\circ\ph_{ji}$ from $v_i$ to $v_j$ via $q_{ij}$ as its embedding.  The resulting path lies in the interiors of $P_i$ and of $P_j$ except at $q_{ij}$ and does not 
	intersect any other polygon.
	See Figure~\ref{fig:ab-graphs} for an example. 
	
	We can compute this planar embedding of $\agraph(A)$ in $O(N(n+\log N))$ time by first computing a 
	triangulation of all $P_i$'s in $O(N\log N)$ time using a simple sweep-line algorithm (or in $O(N)$ time 
	using the linear-time triangulation algorithm~\cite{Ch91}) and then computing geodesics within each $P_i$ 
	in a total of $O(Nn)$ time~\cite{GHL*}.\footnote{More precisely, computing all the geodesics can be done in $O(N_0 n)$ time, where $N_0$ is the maximum complexity of a single polygon in $A$; we use the coarser $O(Nn)$ time for simpler exposition.}
	
	\begin{figure}[hbt]
		\centering
		\begin{tabular}{cccc}
			\includegraphics[page=1, width=0.22\textwidth]{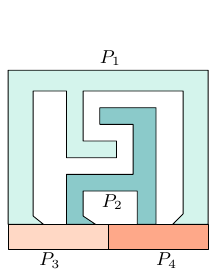} %
			&
			\includegraphics[page=2, width=0.22\textwidth]{figures/assembly_with_graph.pdf} 
			&
			\includegraphics[page=3, width=0.20\textwidth]{figures/assembly_with_graph.pdf}
			&
			\includegraphics[page=4, width=0.20\textwidth]{figures/assembly_with_graph.pdf}
			\\
			(a)&(b)&(c)&(d)
		\end{tabular}
		\caption{(a) An assembly $A$.
			(b) The planar embedding of $\agraph(A)$ overlaid on top of $A$; the $q_{ij}$'s are drawn as small white discs.
			(c) $\agraph(A)$.
			(d) The blocking graph $\bgraph(A)$. For the subassembly $S=\set{P_4}$ we have $CS(S)=\set{P_1, P_2, P_4}$.
		}
		\label{fig:ab-graphs}
	\end{figure}

	\paragraph{Blocking graph}
	The second graph, denoted by $\bgraph(A)$, is the \emph{blocking graph} of $A$. It is a directed graph that 
	represents the collision relation between the polygons of $A$. $\bgraph=\bgraph(A)$ has a directed edge $v_i \rightarrow v_j$ if 
	and only if $P_j$ is visible from $P_i$ in \dirup{}, namely, there is a point $x\in P_i$ such that the ray emanating 
	from $x$ in \dirup{} intersects $P_j$ before intersecting any other polygon of $A$.
	It is not hard to see that $\bgraph$ is planar and thus has $O(n)$ edges. 
	$\bgraph$ can be computed in $O(N\log N)$ time by sweeping a vertical line from $x=-\infty$ to $x=+\infty$, 
	maintaining the intersection of the sweep line with the polygons of $A$, and observing that if $P_j$ is vertically visible from $P_i$ then they will appear next to each other for some value of $x$ during the sweep.
	See Figure~\ref{fig:ab-graphs} for an example.
	
	Our blocking graph is similar to the \emph{directional blocking graph} of Wilson and Latombe~\cite{wilson1994geometric}; we construct a more compact version of the blocking graph, which is sufficient for our needs---we only record immediate, visible, blocking relations, whereas they record all pairs ($P_i,P_j$) where $P_j$ blocks $P_i$.

	For a subassembly $S \subseteq A$, we define the \emph{collision set} of $S$, denoted by $\CS(S)$, to be 
	$$\CS(S) \coloneqq \{P_i \in A \mid v_i \text{ is reachable from } \nodef{S} \text{ in } \bgraph(A)\},$$
	i.e., the set $S$ itself plus the polygons of $A$ that have to be translated upward in order for $S$ to translate upward
	without collisions. Given $S$, $\CS(S)$ can be computed in $O(n)$ time by performing a depth-first search from 
	$\nodef{S}$ in $\bgraph$. The following lemma 
	which follows immediately from the definition of the collision set, states the properties of $\CS(S)$ that we will need.
	\begin{lemma}
		\label{lem:shadow}
		For any subset $S\subseteq A$,
		\begin{enumerate}[label=(\roman*)]
			\item $\CS(S)$ can be moved upward without colliding with $\compl{\CS(S)}$.
			\item If $S^*$ is a connected  partition that contains $S$, then $\CS(S) \subseteq S^*$.
		\end{enumerate}
	\end{lemma}
	
	We compute $\agraph(A)$, including its planar embedding, and $\bgraph(A)$ as a preprocessing step and assume that we have them at our disposal during the algorithm. The total time spent preprocessing is $O(N(n+\log N))$.

	\subsection{Overall algorithm}
	
	We describe an algorithm that given $A$ 
	and a parameter $k \le n/2$ determines, in $O(2^k n^2)$ time,  whether there 
	exists a connected partition $S^* \subset A$ of size at most $k$ in \dirup{}, namely, $S^*$ can be moved arbitrarily far away in \dirup{} without colliding with $\compl{S^*}$ and both $S^*, \compl{S^*}$ are connected.
	If the answer is yes, the algorithm also returns $S^*$.

	Our overall approach is based on the following observation.
	Since $\agraph(A)$ is connected, for any partition $S \subset A$, there is a pair of ``adjacent'' polygons 
	$P_s \in S$, $P_t \in \compl{S}$ such that $(v_s,v_t) \in E$.
	We therefore fix an edge $(v_s,v_t)\in E$  and describe a recursive algorithm, called \algname{}, for computing a 
	connected partition $S^*$ of size at most $k$, with the additional property that $P_s \in S^*$ and~$P_t \notin S^*$, if there exists one. We refer to such an $S^*$ as an \emph{$(s,t)_k$-partition}. \algname{} takes $O(2^k n)$ time.
	By repeating this algorithm for all $O(n)$ edges of $\agraph$, we compute a connected partition of size at most $k$, 
	if there exists one, in $O(2^kn^2)$ time.
	
	\subsection{The \algname{} procedure}
	Each recursive call of \algname{} receives as input a ``partial solution''---a (possibly not connected) subassembly $S \subset A$ of size at most $k$ with $P_s\in S$ and $P_t\not\in S$---along with $P_t$ and $k$.\footnote{The arguments $P_t$ and $k$ remain fixed in all recursive calls of \algname, so for simplicity, we will ignore them, and focus on $S$.}
	The algorithm grows $S$ until it either finds an $(s,t)_k$-partition $S^*\supseteq S$ 
	or determines that there is no $(s,t)_k$-partition containing $S$. 
	It assumes the following property for the input partial solution $S$:
	\begin{itemize}
		\item[(*)] $S$ spans a contiguous portion of the $x$-axis, i.e., the union of $x$-projection of polygons in 
		$S$ is a single interval.
	\end{itemize}
	Initially, \algname{} is called with the partial solution $S_0 \coloneqq \{P_s\}$, which satisfies (*).
	
	Given a subassembly $S \subseteq A$ that satisfies (*),
	each recursive call of \algname{} performs one of the following three actions:
	\begin{itemize}
		\item[(i)] It returns  an $(s,t)_k$-partition $S^*$ containing $S$, in which case the overall algorithm stops and returns $S^*$.
		\item[(ii)] It concludes that there is no $(s,t)_k$-partition containing $S$, in which case the algorithm 
		aborts that recursive call. 
		\item[(iii)] It returns (at most) two subassemblies $S_1, S_2$, each properly containing  $S$ such that any 
		$(s,t)_k$-partition that contains $S$ also contains at least one of $S_1$ or $S_2$. In this case, 
		the algorithm is called recursively with $S_1$ and with $S_2$.
	\end{itemize}
	
	\algname{} performs these actions in the following two steps.
	The first step computes $B \coloneqq \CS(S)$, in $O(n)$ time, by performing a depth-first search on $\bgraph$,
	starting from the 
	vertices of $S$. By Lemma~\ref{lem:shadow}, 
	$B$ is the necessary set of polygons of $A$ required to allow the translation of $S$ to infinity in the upward direction 
	without collisions, resulting in an augmented subassembly $B \supseteq S$.
	If $|B|>k$ or $P_t\in B$, we conclude that there is no $(s,t)_k$-partition containing $S$, and we return \false{}.
	So, assume $|B| \le k$ and $P_t\notin B$.
	If both $B$ and $\compl{B}$ are connected, then $B$ is an $(s,t)_k$-partition  (see Lemma~\ref{lem:shadow}(i)), 
	and we return $B$.
	
	If either $B$ or $\compl{B}$ is not connected, the second step
	finds at most two sets of polygons to add to $B$, such that at least one of these two subsets needs to be added to $B$ to form a connected partition and such that the number of connected components in the disconnected subassembly is 
	reduced. This step proceeds as follows:
	If $\compl{B}$ is not connected, then let $T_1, \ldots, T_r$ be its connected 
	components that do not contain $P_t$.
	These components must be added to $B$ for $\compl{B}$ to be connected. If $|B|+\sum_i |T_i|>k$ then there is 
	no $(s,t)_k$-partition containing $B$, and we return \false{}.
	Otherwise, we recursively call \algname{} with
	$B \cup T_1 \cup \cdots \cup T_r$, which concludes this case.
	
	Next, we assume that $\compl{B}$ is connected but $B$ is not connected.
	In this case we call the subroutine \subprocname{}$(B)$, which, in $O(n)$ time, identifies at most two candidate subassemblies $B \subset S_1, S_2 \subseteq A$, such that each of $S_1, S_2$ has fewer connected components than $B$, and 
	any connected partition $S^*$ containing $B$ contains at least either $S_1$ or $S_2$ (cf.\ Lemma~\ref{lem:connect}).
	For $i=1,2$, if $|S_i| \le k$ and $P_t\not\in S_i$, we recursively call \algname{} with on $S_i$. Otherwise 
	($|S_i|>k$ or $P_t\in S_i$ for both $i$), there is no $(s,t)_k$-partition 
	containing $B$ and we return \false{}.
	
	This concludes the description of the \algname{} procedure.
	For an integer $0 \le x \le k$, let $\tau(x)$ be the maximum running time of 
	$\algname$ on a partial solution of size $k-x$. 
	For $x=0$, the procedure spends $O(n)$ time, as no recursive call is made.
	For $x>0$, it spends $O(n)$ time on non-recursive processing and makes at most two recursive calls, each
	with a partial solution of size larger than $k-x$. Hence, we obtain the following recurrence,
	\[
	\tau (x) \le \left \{ \begin{array}{ll}
		2\tau(x-1) + O(n) & \mbox{for $x>0$,}\\[1mm]
		O(n) &  \mbox{for $x=0$.}
	\end{array}
	\right .
	\]
	The solution to the above recurrence is easily seen to be $O(2^{x} n)$, as the depth of the recursion is at most 
	$x$. Since $x=k-1$ for the initial call of \algname, we obtain the following:
	
	\begin{lemma}
		\label{lem:augment}
		Given a pair $(P_s,P_t)$ and an integer $k>0$, \algname{} computes, in $O(2^k n)$, time an $(s,t)_k$-partition,
		if one exists.
	\end{lemma}
	
	\subsection{The \subprocname{} procedure}
	
	Next, we describe the \subprocname{} procedure.
	It receives a subassembly $B \subseteq A$ that is not connected and has the following properties: 
	\begin{enumerate}
		\item[(C1)] $B$ satisfies (*), 
		\item[(C2)] 
		$B$ can be translated upward without colliding with $\compl{B}$, and \label{prop:B_upwards} %
		\item[(C3)] $\compl{B}$ is connected.
	\end{enumerate}
	It computes at most two candidate subassemblies $B \subset S_1, S_2 \subseteq A$, such that each of 
	$S_1, S_2$ has fewer connected components than $B$, and any $(s,t)_k$-partiton contains at least one 
	of $S_1$ and $S_2$.
	Before describing the \subprocname{} procedure in detail, we introduce a few concepts and prove a few structural properties of $\agraph$.

	\paragraph{Shadow of a subassembly}
	For a subassembly $S \subseteq A$, we define the \emph{shadow} of $S$, denoted by $\Sh(S)$, to be the union of all rays in \dirup{} that emanate from points in the polygons of $S$, namely,
	\[ \Sh(S) = \bigcup_{P\in S}\bigcup_{p\in P} \{ p+ (0,\lambda) \mid \lambda\ge 0\} .\]
	We define $L(S)$ to be the boundary of $\Sh(S)$.
	If $S$ satisfies (*) then $L(S)$ is an $x$-monotone polygonal chain consisting of: (i) the lower envelope 
	of the edges of the polygons in $S$ (i.e., portions of the edges that are visible from $y=-\infty$), 
	(ii) vertical segments connecting consecutive breakpoints of the lower envelope,
	and (iii) two rays in \dirup{} emanating from the leftmost and the rightmost endpoints of the lower envelope.
	See Figure~\ref{fig:shadow}.
	We set $L \coloneqq L(B)$ and orient $L$ in counterclockwise manner so that $\Sh(B)$ lies to the left of $L$ and 
	$\compl{B}$ lies to the right of $L$.
	Since $B$ can be translated upward without colliding with $\compl{B}$, $L$ does not intersect the interior of 
	any polygon of~$A$.
	
	\begin{figure}[htb]
		\centering
		\begin{tabular}{cccc}
			\includegraphics[page=5, width=0.22\textwidth]{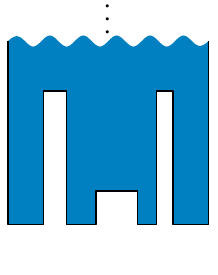}
			&
			\includegraphics[page=6, width=0.22\textwidth]{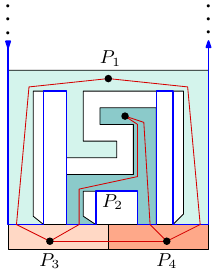} 
			&
			\includegraphics[page=7, width=0.22\textwidth]{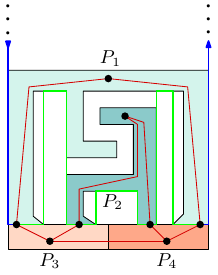}
			&
			\includegraphics[page=8, width=0.2\textwidth]{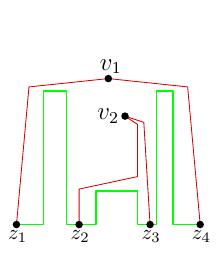}
			\\
			(a)&(b)&(c)&(d)
		\end{tabular}
		\caption{The examples in this figure are based on the subassembly $B \coloneqq \set{P_1, P_2}$ of the assembly in Figure~\ref{fig:ab-graphs}.
			(a) Shadow $\Sh(B)$ ;
			(b) its boundary $L(B)$ (thick blue), shown with $A$ and $\agraph(A)$;
			(c) the augmented graph $\rgraph$, which has three new edges along $L(B)$ (thick green), shown with $A$ and $L(B)$;
			(d) and the subgraph $\ograph$.}
		\label{fig:shadow}
	\end{figure}
	
	\paragraph{Refining the adjacency graph}
	We now define a plane graph $\rgraph$ that we obtain by splitting some of the edges of $\agraph$ and adding some edges so that, roughly speaking, $L$ becomes part of the graph:
	Let $(v_i,v_j)$ be a ``cut'' edge of the subassembly $B$, i.e., $P_j \in B$ and $P_i\in \compl{B}$. By construction, 
	$q_{ij}$, the point added to $\bd P_i\cap \bd P_j$ to compute the embedding of $(v_i,v_j)$, lies on $L$. 
	We add $q_{ij}$ as a vertex of $\rgraph$ and replace the edge $(v_i,v_j)$ in $\agraph$ by two edges $(v_i, q_{ij})$ and $(v_j,q_{ij})$, with $\ph_{ij}, \ph_{ji}$, respectively, being their embeddings. 
	Let $Z=\set{z_1,\ldots,z_m}$ be the sequence of the new nodes (i.e., $q_{ij}$'s) added to $\rgraph$ sorted in counterclockwise order along $L$.
	For each $1 \le i < m$, we also add the edge $e_i \coloneqq (z_i, z_{i+1})$ to $\rgraph$, with the portion of $L$ 
	between $z_i$ and $z_{i+1}$ being its embedding.
	For $1\le a<b\le m$, let $L[a,b] \coloneqq \langle z_a, z_{a+1}, \ldots, z_b\rangle$ denote the path in $\rgraph$ 
	from $z_a$ to $z_b$ along $L$.
	Clearly, $\rgraph$ is also a simple planar graph with $O(n)$ vertices and edges.
	The planar embedding of $\rgraph$ is the same as that of $\agraph$ plus the portion of $L$ between $z_1$ and $z_m$.
	See Figure~\ref{fig:shadow} for an example of $\rgraph$ and its planar embedding.
	
	Since $B$ is not connected, 
	let $B_1, \ldots, B_u \subset B$, for $u>1$, be the connected components of $B$.
	Let $Z_i = \set{z_j \mid z_j \in \bd P \text{ for some } P \in B_i}$ be the subsequence of the newly added vertices that 
	lie in $B_i$.
	The assembly $A$ being connected implies that $Z_i \ne \emptyset$ for every $i=1,\ldots, u$, as at least one vertex in $\nodef{B_i}$ is connected to $\nodef{\compl{B}}$.
	The following two properties of the sets $Z_i$ will be crucial for our algorithm.
	
	\begin{lemma}
		\label{lem:nested}
		The nodes in $Z_1, \ldots, Z_u$ form a nested structure, i.e., there are no six indices $i< j \le u$  and $a < b < c < d$ such that $z_a, z_c  \in Z_i$ and $z_b, z_d \in Z_j$.
	\end{lemma}
	\begin{proof}
		For a contradiction, suppose such a set of indices exists.
		Let $\ograph$ be the subgraph of $\rgraph$ induced by the vertices of $\nodef{B}\cup Z$; $\ograph$ lies in 
		$\Sh(B)$ and $\agraph(B)$ is a subgraph of $\ograph$; see Figure~\ref{fig:shadow}~(d).
		For a vertex $z_p \in Z$, let $p'$ be the index of the polygon of $B$ whose boundary contains $z_p$, 
		i.e., $z_p \in \bd P_{p'}$.  
		Since $v_{a'}$ and $v_{c'}$ lie in the same connected component of $\agraph(B)$, there is a path $\pi'_i$ between them.
		Let $\pi_i := (z_a,v_{a'}) \circ \pi'_i \circ (v_{c'},z_c)$ be the path from 
		$z_a$ to $z_c$ in $\ograph$; by construction $\pi_i$ does not use any vertex of $Z \setminus \{z_a,z_b\}$. 
		Similarly define the path $\pi_j$ from $z_{b}$ to $z_{d}$ in $\ograph$.
		Since (the planar embeddings of) $\pi_i$ and $\pi_j$ are simple polygonal curves lying inside $\Sh(B)$ with their 
		endpoints on $L$ and alternating along $L$, they cross each other. Furthermore 
		$\ograph$ is planar, therefore $\pi_i$ and $\pi_j$ must share a vertex, which, by construction, is 
		a vertex of $\nodef{B}$.  However, this contradicts the fact that 
		$B_i$~and~$B_j$ are two different connected components of $B$.
	\end{proof}

	For each $i\le u$, let $\ell_i$ (resp.\ $r_i$) be the index of the first (resp.\ last) vertex of $Z_i$ along $L$.
	We say that $Z_i$ \emph{appears before} $Z_j$ if $\ell_i < \ell_j$.
	Let us fix the indices of the sets $Z_i$ so that they are ordered according to $\ell_i$, namely, 
	$\ell_i < \ell_j$ for $i<j$.
	Next, we show that $Z_u$, the set that appears last, forms a contiguous subsequence in $Z$, namely no other $Z_i$ shows up between $z_{\ell_u}$ and $z_{r_u}$ along $L$.
	
	\begin{lemma}
		\label{lem:cont}
		$Z_u = \{ z_i \mid \ell_u \le i \le r_u \}$.
	\end{lemma}
	\begin{proof}
		Suppose to the contrary there is an index $i$ such that $\ell_u < i < r_u$ and $z_i \in Z_j$ for some $j < u$, i.e., $\ell_j < \ell_u$.
		We now have four indices $\ell_{j} < \ell_u < i < r_u$ such that $z_{\ell_{j}}, z_i \in Z_{j}$ and $z_{\ell_u}, z_{r_u} \in Z_u$, contradicting Lemma~\ref{lem:nested}.
	\end{proof}
	
	\begin{lemma}
		\label{lem:bounded}
		For any edge $e_i = (z_i, z_{i+1})$ in $\rgraph$, the face adjacent to $e_i$ and lying to its right 
		(i.e., not lying in $\Sh(B)$) is bounded.
	\end{lemma}
	\begin{proof}
		Let $F$ be the face adjacent to $e_i$ lying to its right.
		Let $q_{ab}$ and $q_{cd}$ denote the endpoints of $e_i$, where $P_a, P_c \in \compl{B}$ (and possibly $a=c$).
		By definition, $\rgraph$ contains the edges $e'=(v_a, q_{ab})$ and $e''=(v_c, q_{cd})$.
		Since $\agraph(\compl{B})$ is connected subgraph of $\rgraph$, there is a path $\gamma$ 
		between $v_a$ and $v_c$ in $\rgraph$ that lies to the right of $L$.
		Therefore, $e' \circ \gamma \circ e'' \circ e_i$ is a cycle in $\rgraph$ that lies on or to the right of $L$. 
		This cycle bounds the face $F$ because $F$ contains $e_i$ and also lies to the right of $L$. See Figure~\ref{fig:bounded_faces}.
	\end{proof}

	\begin{figure}[H]
		\centering
		\includegraphics[width=0.32\textwidth]{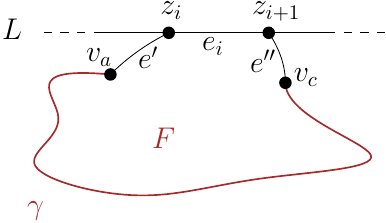}
		\caption{A partial view of $\rgraph$ in an embedding in which $L$ is stretched out on a line, illustrating Lemma~\ref{lem:bounded}.
		Here $z_i = q_{ab}$ and $z_{i+1} = q_{cd}$.
		}
		\label{fig:bounded_faces}
	\end{figure}

	\paragraph{Computing $S_1$ and $S_2$}
	We are now ready to describe the  \subprocname{} procedure, which
	chooses additional polygons to be added to~$B$.
	Given $B$, we first compute $\rgraph$ and the sets $Z_1, \ldots Z_u$. 
	
	Next, we identify two edges $e_<$ and $e_>$ in $\rgraph$, as follows.
	Recall that $u>1$, so $\ell_u > 1$.  We define the edge $e_<$ to be $e_{\ell_u - 1} = (z_{\ell_u-1}, z_{\ell_u})$ of $\rgraph$.
	If $Z_u$ does not contain the last point of $Z$ along $L$ (i.e., $r_u < m$), then we define $e_> := e_{r_u} = (z_{r_u}, z_{r_u+1})$; otherwise ($r_u=m$) $e_>$ is undefined. 
	
	For concreteness, we assume that both $e_{<}$ and $e_{>}$ exist, otherwise the situation is simpler.
	We orient $e_{<}, e_{>}$ in counterclockwise direction along $L$.
	Let $F_{<}$ (resp.\ $F_{>}$) be the face of $\rgraph$ adjacent to $e_{<}$ (resp.\ $e_{>}$) lying to its right.
	By Lemma~\ref{lem:bounded}, both faces are bounded.

	Consider $\bd F_{<}$, the boundary of the face $F_{<}$. $\bd F_<$ is a non-self-crossing tour in~$\rgraph$; it need not 
	be simple, and  edge $e_{<}$ is part of $\bd F_{<}$. 
	Define $S_< \subseteq \compl{B}$ to be the set of polygons corresponding to those nodes of 
	$\bd F_{<}$ that do not lie on $L$ (the endpoints of $e_<$ are the only nodes that lie on $L$).
	We define $S_{>}$ analogously for $F_{>}$.
	We return the sets $S_1 \coloneqq B \cup S_>$ and $S_2 \coloneqq B\cup S_<$.  The running time of
	\subprocname{} is $O(n)$.
	
	\paragraph{Proof of correctness}
	It can be verified that $S_1$ and $S_2$ satisfy the invariant~(*), so we
	now prove that any $(s,t)_k$-partition containing $S$ contains either $S_1$ or $S_2$.

	\begin{lemma}
		\label{lem:cycle}
		Let $\pi$ be a simple path in $\rgraph$ whose endpoints $z_a, z_b$ lie in $Z$ and the rest of the nodes lie 
		in $\nodef{\compl{B}}$ and let $\Delta$ be the region bounded by the cycle $\pi\circ L[a,b]$ (in $\rgraph$).
		Let $\compl{B}_\pi \subseteq \compl{B}$ be the set of polygons 
		corresponding to the nodes of $\pi$ except its endpoints $z_a, z_b$, and let 
		$S \subseteq \compl{B}$ be the set of polygons
		corresponding to the nodes of $\agraph(\compl{B})$ that lie in $\Delta$. 
		Then $S \subseteq \CS(\compl{B}_\pi)$.
	\end{lemma}
	
	\begin{proof}
		Let us fix a polygon $P_j \in S$. Then the point $v_j\in P_j$ lies inside the polygonal cycle formed by (the planar embedding of) $\pi \circ L[a,b]$.
		The ray $\rho$ emanating from $v_j$ in \dirdown{} intersects $\pi$; $\rho$ does not intersect $L[a,b]$ because 
		$v_j$ lies below the $x$-monotone chain $L$. Every point in $\pi$ lies inside a polygon of $\compl{B}_\pi$, so 
		$\rho$ intersects a polygon $P_i \in \compl{B}_\pi$, which implies that $P_i$ collides with $P_j$ if we move 
		$P_i$ in \dirup. Hence, $P_j \in \CS(\compl{B}_\pi)$, as claimed.
	\end{proof}

	\begin{lemma}
		\label{lem:minimal}
		Let $S^*$ be any $(s,t)_k$-partition that contains $B$, then $B \cup S_{<} \subseteq S^*$ or $B \cup S_{>} \subseteq S^*$.
	\end{lemma}
	\begin{proof}
		Recall that $B_u$ is the last connected component of $B$ to appear along $L$.
		Since $S^*$ is connected, $\agraph(S^*)$ must contain a simple path from $\nodef{B_u}$ to 
		$\nodef{B_j}$, for some $j < u$.
		Among all such paths, let $\pi=\langle v_1, v_2, \ldots, v_r\rangle$ be a minimal one, i.e., no subpath of 
		$\pi$ connects $B_u$ to 
		$\nodef{B}\setminus \nodef{B_u}$. Then $v_1 \in \nodef{B_u}$, $v_r\in \nodef{B_j}$, and 
		$v_2, \ldots, v_{r-1} \in \nodef{\compl{B}}$.
		We set $\hat\pi \coloneqq \langle q_{12}, v_2, \ldots, v_{r-1}, q_{(r-1)r}\rangle$, which is a path
		in $\rgraph$ from a vertex $z_a \in Z_u$ to~$z_b \in Z_j$, with
		$a \in [\ell_u,r_u]$ and $b \notin [\ell_u, r_u]$ (by Lemma~\ref{lem:cont}).  
		Assume that $b > r_u$, i.e., $\ell_u \le a \le r_u < b$. Let $\compl{B}_{\hat\pi}\subseteq\compl{B}$ be the 
		the set of polygons 
		corresponding to the nodes of $\hat\pi$ except its endpoints $z_a, z_b$ (same as defined in Lemma~\ref{lem:cycle}). By definition $\compl{B}_{\hat\pi} \subset S^*$.
		
		Let $\Delta$ be the region bounded by $\hat\pi\circ L[a,b]$. $\bd\Delta$ is a cycle in $\rgraph$ lying on or 
		to the right of~$L$ and~$e_{>} \in \bd\Delta$ (because $a \le r_u < b$).  Since $F_{>}$ is a face 
		of $\rgraph$, $\bd F_{>}$ lies in~$\Delta$, possibly overlapping $\bd\Delta$, 
		which implies that the polygons of $S_{>}$ lie in~$\Delta$; see Figure~\ref{fig:minimality}.
		Therefore, by Lemma~\ref{lem:cycle}, $S_{>} \subseteq \CS(\compl{B}_{\hat\pi})$.  
		Since $\compl{B}_{\hat\pi} \subset S^*$, by Lemma~\ref{lem:shadow}(ii), 
		$\CS(\compl{B}_{\hat\pi}) \subseteq S^*$.
		Therefore $S_{>} \subseteq S^*$, as required.
		
		Finally, if $b < r_u$, we similarly conclude that $S_{<} \subseteq S^*$.
	\end{proof}
	\begin{figure}[htb]
		\centering
		\includegraphics[width=0.45\textwidth]{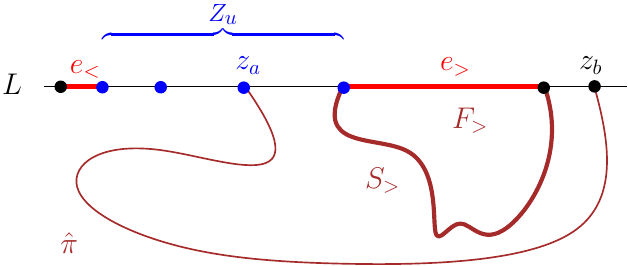}
		\caption{A partial view of $\rgraph$ in an embedding in which $L$ is stretched out on a line. The $z_i$'s are dots on $L$ with the nodes of $Z_u$ in blue. Also shown: edges $e_<$ and $e_>$ (red segments), face $F_>$, $S_>$ (thick curve), and path $\hat\pi$ from the proof of Lemma~\ref{lem:minimal}.
		}
		\label{fig:minimality}
	\end{figure}
	
	Putting everything together, we obtain the following:
	\begin{lemma}
		\label{lem:connect}
		Given a subset $B\subset A$ that satisfies (C1)--(C3),
		\subprocname{} computes, in $O(n)$ time, two sets $B \subset S_1, S_2 \subseteq A$
		such that any 
		connected partition containing $B$ contains at least one of them. Furthermore $S_1, S_2$ satisfy the property (*).
	\end{lemma}

\section{Linear-Time Partitioning Algorithm for Horizontally Monotone Grid Square Assemblies}
	\label{sec:positive_res}
	In this section, we present a restricted case of \textsc{PPCST-GRID}, for which a connected partition always exists, and show how to obtain it in linear time.
	Let $A$ be a connected \gridass{} with $n$ squares.
	A directed path of squares $s_1, \ldots, s_m$ in $\mathcal{G}(A)$ is called \emph{rightward (resp. leftward) monotone} if for all $i<m$, $s_{i+1}$ is not located to the left (resp. right) of $s_i$.
	$A$ is \emph{\monass} if either one of the following holds: (i) from every square of $A$ there is a rightward monotone path to some rightmost square of~$A$ or (ii) from every square of $A$ there is a leftward monotone path to some leftmost square of~$A$.
	For convenience, we assume in this section that a \monass{} assembly satisfies condition (i).
	
	Checking whether a given \gridass{} $A$ is \monass{} can be done in $O(n)$ time as follows.
	Let $R \subseteq A$ denote all the squares that are rightmost in their grid row, which can be found in linear time.
	Run a breadth-first search (BFS) from $R$ in $\mathcal{G}(A)$ that does not visit the right neighboring square of the currently visited square (i.e., this BFS only considers leftward monotone paths from $R$).
	It is easy to verify that $A$ is \monass{} if and only if this BFS visits all the squares in $A$.

	\begin{theorem}
		A connected \monass{} \gridass{} with $n>1$ squares can always be partitioned into two connected subassemblies by a single vertical translation. Such a partition can be found in $O(n)$ time.
		\label{thm:solvable}
	\end{theorem}

	\begin{proof}
		\newcommand{\amins}{\ensuremath{A \setminus S}}
		\newcommand{\aminC}{\ensuremath{A \setminus C}}
		We show how to obtain a connected subassembly $S \subset A$ that can be rigidly translated upward without colliding with $A \setminus S$, such that $A \setminus S$ is also connected.
		Let $\ell$ denote the top square on the leftmost column of $A$ (namely the leftmost column of the grid that contains a square of $A$).
		If $A\setminus\set{\ell}$ is connected, then $\set{\ell}$ is a valid choice for $S$, since both subassemblies are connected and there is nothing blocking a translation of $\ell$ upward to infinity, and we are done.
		Otherwise, let $C$ be the top connected component of the leftmost column of $A$.
		If $A\setminus C$ is connected, then as before, we are done.
		If that is not the case, then let $D$ denote the set of squares of $A$ that lie in the column adjacent to the leftmost column of $A$.
		Since $A\setminus C$ is not connected, there are at least two squares in $D$ that are adjacent to $C$.
		Among all such squares, let $s$ and $s'$ be the two topmost squares, where $s$ is above $s'$; see Figure~\ref{fig:solvable}.
		The square $s$ must be adjacent to $\ell$, as otherwise $A\setminus\set{\ell}$ is connected, and $\set{\ell}$ is a valid choice for $S$.
		By the same argument, $s'$ cannot be adjacent to $s$.
		
		Let $S$ be the connected component of $\mathcal{G}(A\setminus\set{\ell})$ containing $s$ and let $S'=A\setminus S$, which is also a connected subassembly. We have $s' \in S'$, as otherwise $s' \in S$ and then $s$ and $s'$ must lie on a cycle in $\mathcal{G}(A)$ that $\ell$ is also part of.
		This makes $\set{\ell}$ a valid choice for $S$, which we assumed is not the case.
		We claim that $S$ can be translated upward without colliding with $S'$.
		We therefore assume to the contrary that there is a square $t \in S$ that is located below a square $t' \in S'$ in the same grid column.
		Since $A$ is \monass{}, there is a rightward monotone path from each of $t$ and $t'$ to some rightmost square in $A$.
		Let $P$ and $P'$ denote these respective paths and let $r, r' \in R$ be their respective endpoints.
		The square $r$ must be located on a lower row than the row $r'$ is in, as otherwise $P$ and $P'$ are not disjoint, which would contradict the disjointness of $S$ and $S'$.
		Let $Q \subseteq S$ be the union of $P$ and the path from $s$ to $t$ that exists in $S$ and let $Q' \subseteq S'$ be defined similarly for $P',s',t'$.
		$Q'$ cannot contain any square that is to the right of $r$ on the same grid row or any square that is to the left of $C$ (as such squares are not part of $A$).
		Therefore, the only way for $Q'$ to connect $s'$ with $r'$ is to go through $Q$, which contradicts the disjointness of $S$ and $S'$.
		In conclusion, $S$ can be translated upward to infinity without colliding with $S'$, and both $S$ and $S'$ are connected.
		
		\begin{figure}[htb]
			\centering
			\includegraphics[width=0.26\textwidth]{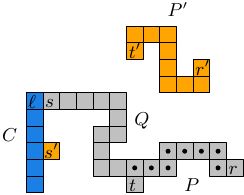}
			\caption{The final contradiction arising in the proof of Theorem~\ref{thm:solvable}, where there is no way to connect $s'$ and $r'$ in $S'$.
				$C$ is blue, $Q$ is gray, and squares belonging to $P$ are marked by dots.}
			\label{fig:solvable}
		\end{figure}
		
		Computing such a partition involves identifying $\ell$ and $C$ and computing the connected components of $\mathcal{G}(A\setminus\set{\ell})$ and $\mathcal{G}(A\setminus{C})$, all of which take $O(n)$ time.
		If the adjacency graph $\mathcal{G}(A)$ is not given, we construct it as follows.
		Since $A$ is connected, we can normalize it to lie in the square subgrid $[0,n]\times[0,n]$ and then use two lexicographic radix sorts to construct $\mathcal{G}(A)$ in $O(n)$ time.
		
	\end{proof}

\section{Conclusion}
We have shown that given a planar assembly, finding a connected partition for a given direction is hard.
In assembly planning one may not require the separating motion to occur along a specific direction, but rather ask whether a connected partition exists for any direction.
This problem remains NP-complete since our construction for \textsc{PPCST} can be easily modified so that a partition can only occur along the vertical direction (e.g., by having the positive root clause be constrained to move up due to sliding contact with the negative root clause).
The construction for \textsc{PPCST-GRID} can be similarly made to only have a partition in the vertical direction. This requires modifying the quasi-polygon $B$ and its symmetric counterpart (located below the variable row) to emulate two rigid parts that enclose the construction and constrain each other's motion.
\textsc{PPCST-GRID} remains NP-complete in three dimensions, in which case the parts are unit cubes each occupying a cell of the spatial unit grid, since we can give the squares in our construction volume along the $z$-axis.
However, the complexity of deciding whether a connected partition exists for any direction for unit-grid cubes remains open, since it does not immediately follow from the aforementioned construction. %
The same question (i.e., arbitrary direction) for general polyhedral parts is NP-complete, because we can again modify the construction for \textsc{PPCST} to restrict the only possible partition direction to be \dirup{}.

Our results indicate that the overall shape of the assembly, i.e., the shape of the union of its constituent parts (as opposed to the shape of individual parts), plays an important role in the difficulty of the connected-partitioning problem.
For example, "hooks" are a crucial part of our gadgets in both the polygonal and square construction;
we showed that the gadgets' functionality can be kept the same regardless of whether the hooks are rigid or made of squares that resemble a hook.
This suggests that more results can be obtained by focusing on the implications of the overall shape of the assembly, as we explored for \monass{} assemblies in Section~\ref{sec:positive_res}.

We remark that an FPT algorithm for connected-assembly-partitioning does not solve the original motivating problem of finding a complete (dis)assembly sequence with the connectivity constraint.
This is true because it is possible to have a subassembly resulting from a connected partition to not have a connected partition of its own (see, e.g.,~\cite[Figure 4]{DBLP:journals/ral/ManzoorSLKKB17}).
Therefore, finding a complete sequence with the connectivity constraint requires further investigation.

\section*{Acknowledgments}
We thank Nancy M. Amato,
Stav Ashur,
Gali Bar-On,
Gil Ben-Shachar,
Kostas Bekris,
J.~Frederico Carvalho,
Erik Demaine,
Jayson Lynch, and
Jay Tenenbaum
for useful discussions.
We also thank the reviewers for their helpful comments on the paper.

\bibliographystyle{siamplain}
\bibliography{references}

\end{document}